\documentclass[aps,prx,reprint,superscriptaddress,longbibliography,showkeys]{revtex4-2}

\usepackage{amsmath,amssymb,mathtools,bm}
\usepackage{amsthm}
\usepackage{xcolor}
\usepackage{graphicx}
\usepackage{microtype}
\usepackage{enumitem}
\usepackage{booktabs}
\usepackage{ragged2e}
\usepackage{url}
\usepackage[hidelinks]{hyperref}

\newcommand{\ket}[1]{\left| #1 \right\rangle}
\newcommand{\polylog}[1]{\textrm{polylog}\left( #1 \right)}
\newcommand{\bra}[1]{\left\langle #1 \right|}
\newcommand{\ssb}[1]{{\color{blue}#1}}
\newcommand{\sgn}[1]{\textrm{sgn}\left(#1 \right)}
\newcommand{\braket}[2]{\left\langle #1 | #2 \right\rangle}

\newtheorem*{problem*}{Problem Statement}
\newtheorem{theorem}{Theorem}
\newtheorem{remark}{Remark}
\newtheorem{corollary}{Corollary}

\newtheorem{lemma}{Lemma}
\newtheorem{proposition}{Proposition}

\providecommand{\Endparasplit}{}
\providecommand{\showacknow}[1]{}

\newcommand{\dataavail}[1]{\section*{Data Availability}#1}
\newcommand{\acknow}[1]{\begin{acknowledgments}#1\end{acknowledgments}}

\newcommand{\significancestatement}[1]{%
  \par\begingroup\small\noindent\textbf{Significance Statement.} #1\par\endgroup\vspace{0.75em}}

\begin{document}

\title{Accelerated search for global extrema in structured quantum signals}

\author{Sachin S. Bharadwaj}
\affiliation{Department of Mechanical and Aerospace Engineering, New York University, New York 11201 USA}
\author{Katepalli R. Sreenivasan}
\email{krs3@nyu.edu}
\affiliation{Department of Physics and Courant Institute of Mathematical Sciences, New York University, New York, New York 10012, USA}

\date{This manuscript was compiled on \today}

\begin{abstract}
Extreme values of a physical system are rare by definition, yet they can govern some of its most consequential and anomalous behavior. Their sparsity makes them particularly costly to identify in large data sets. Quantum computation offers a potential way to reduce this cost. The challenge is that quantum amplitudes cannot simply be read, copied, compared, or marked using conventional search oracles. Here we introduce an algorithm for searching amplitude-encoded quantum signals that exploits their underlying structure. Through a nonlinear amplitude transformation, the algorithm identifies and certifies a global extremum while exponentially reducing the effective search space for structured signals. We establish a polylogarithmic query complexity in the data size \(N\), together with global completeness, efficient classical verification, and a measurable criterion for quantifying exploitable structure. We demonstrate the framework using computational turbulence data. Applied to line cuts of velocity gradients from a \(1024^3\) simulation containing more than one billion grid points, the algorithm consistently isolates the extreme tail and reduces the median search space to \(O(\log N)\). The clustered structure characteristic of turbulence further accelerates the search. These results establish a direct connection between physical structure and quantum advantage in extreme-value search. For accessibility, the main text emphasizes the central ideas and results, while the theoretical derivations are provided in the Appendices.

\end{abstract}


\maketitle

\significancestatement{Rare and intense events contain disproportionately important information about complex physical systems. They arise in phenomena as diverse as great floods and earthquakes, extreme volatility in financial markets, and localized hot spots in combustion. Identifying such events in large data sets is challenging in any setting, especially when the data are encoded as amplitudes of a quantum state: the amplitudes cannot be read out and compared without destroying the state. We develop a quantum algorithm that exploits the underlying physical structure of the data to identify such extremes, and illustrate the basic ideas using turbulent velocity gradients. For a nonlinear physical problem the method identifies extreme events with a provable and verifiable end-to-end quantum advantage.}

In many complex physical systems, extreme events can matter far more than averages. Although rare, they may dominate important physical consequences and strongly influence observed scaling and intermittency. Detecting such events is therefore a central challenge, particularly when they are sparsely distributed across large spatiotemporal domains. For example, detecting fast radio bursts from deep space \cite{rajwade2024needle} requires sensors to continuously process enormous volumes of raw data to identify signals that may last only milliseconds. Similar challenges arise in the detection of extreme events in intermittent turbulent flows \cite{Sreenivasan1988,buaria2019extreme}, medical imaging, earthquake monitoring, and tsunami detection \cite{baur2021autoencoders,mousavi2019cred,mungov2013dart}.

The rarity of extreme events makes their detection and statistical characterization computationally expensive, since a large fraction of the data may need to be searched to find a very small fraction of the events of interest. This challenge becomes especially acute for data encoded in quantum states. Quantum amplitudes cannot be directly accessed, copied, or compared repeatedly without disturbing the state, precluding a straightforward application of conventional search strategies. The question is therefore whether the physical structure inherent in many signals can be exploited to make extreme-value searches substantially more efficient \cite{li2008public,lestang2018computing}.

For an unstructured, unsorted database of size $N$, a classical search requires $\mathcal{O}(N)$ queries in the worst case because the entries must, in general, be examined individually. Quantum search algorithms based on Grover's search \cite{grover1996fast} and the D\"urr--H{\o}yer minimum-finding procedure \cite{durr1996quantum,giri2017review}, which we refer to collectively as Grover--D\"urr--H{\o}yer (GDH), reduce the query complexity to $\mathcal{O}(\sqrt{N})$ under the assumptions summarized in Appendix~\ref{sec:gdh}.

\Endparasplit

The prospect of a larger {end-to-end quantum advantage} arises when the data are generated, represented, and processed coherently as quantum states, rather than being first measured as a classical data set and subsequently reloaded into a quantum computer. We therefore consider extrema search directly on an amplitude-encoded signal,
\[
\ket{\Psi}=\sum_{i=0}^{N-1}u_i\ket{i},
\qquad
\sum_{i=0}^{N-1}u_i^2=1,
\qquad
u_i\in\mathbb{R}.
\]
This setting is natural in scientific quantum computing, where many quantum algorithms produce solutions directly in amplitude-encoded form. Quantum algorithms for partial differential equations and linear systems \cite{costa2019quantum,liu2021efficient,krovi2023improved,bharadwaj2023hybrid,bharadwaj2024compact,costa2025carleman,gonzalez2024quantum,bharadwaj2025homotopy,morales2026quantum}, as well as variational algorithms \cite{cerezo2021variational,ingelmann2024two,pool2024nonlinear,kocher2025}, prepare solution states in this form. Coherent quantum-sensor interfaces may likewise provide a route for transducing physical fields directly into quantum states \cite{allen2025quantum,khan2025quantum}.

Our model requires substantially less oracle access than GDH. We assume only access to a state-preparation circuit $U_S$, its inverse $U_S^\dagger$, and their controlled versions. In particular, we do not assume a value oracle, comparison or sign oracle, marking oracle, externally supplied threshold, or prior knowledge of the extremum. Instead, the relevant threshold is determined from the input state itself, and the amplitudes may have either sign. The output is consequently the {true global extremum}, rather than a probabilistically selected candidate. The precise distinctions, together with the physical models for state preparation, are given in Appendices~\ref{sec:gdh} and~\ref{sec:access}.

The key step is a Chebyshev-walk transformation that applies a bounded nonlinear function directly to the amplitudes. Rather than resolving the small numerical differences among the largest few amplitudes, the transformation concentrates probability into a short tail containing only $\mathcal{O}(\log N)$ candidates. High-precision comparison is performed only after this compression. The resulting framework holds for common extrema-search applications.

An efficient search procedure is useful only if the input state can itself be prepared efficiently. This requirement can be met using our previous work. The initial state can be generated by the quantum homotopy solver \cite{bharadwaj2025homotopy}, while derived quantities such as gradients can be constructed coherently using compact LCU matrix operations \cite{bharadwaj2024compact}. The extrema-search algorithm then operates directly on the resulting quantum state, without requiring reconstruction of the full field. The reduction is designed so that the global extremum is retained with certainty under the stated assumptions, while the resulting compact candidate set permits efficient classical verification, including verification in the untrusted-device setting of Ref.~\cite{mahadev2018classical}. Thus, under the stated assumptions, the complete workflow can proceed from nonlinear quantum state generation to a globally correct and classically verifiable observable {without intermediate readout of the full field}.

Our specific fluid dynamic application is the search for extreme turbulent velocity gradients. These are signed, multiscale signals in which rare, intense events are embedded in a broad fluctuating background \cite{Sreenivasan1988}. Their extremes become increasingly intense and intermittent with increasing Reynolds number. Efficiently locating them is therefore both a fundamental problem in turbulence and a stringent test of quantum extrema search.

We combine high-resolution numerical data from homogeneous, isotropic turbulence with analytical guarantees for the algorithm. {The classically available data are first encoded into quantum amplitudes. This step is
needed only to stress test extrema searches on realistic data and is not an intrinsic overhead when the input state is generated quantumly. We quantify the retention probability of high-amplitude events, the performance of finite-degree filtering, candidate-set compression, certification, and the associated computational complexity and error bounds. The analysis also identifies a measurable {structure factor} that characterizes when intermittency concentrates sufficient normalized probability to enable an efficient search. The central result is that, when physically important events are sufficiently rare yet sufficiently concentrated, {their structure can itself serve as the resource for quantum advantage}.

After a statement of the problem immediately below in Sec.~\ref{sec:problem_statement}, we provide an accessible summary of the main results in Sec.~\ref{sec:main_results}. This section lists the lemmas and theorems used in the work, while their proofs, as well as underlying assumptions, are relegated to the Appendices to avoid distraction from the main theme. Sec.~\ref{sec:how_algorithm} provides a summary of how the algorithm works, while Sec.~\ref{sec:numerics} applies the method to turbulent velocity gradients obtained from direct numerical simulation of the Navier--Stokes equations. Sec.~\ref{sec:discussion} provides a discussion of the paper, while Sec.~\ref{sec:conclusions} concludes it with a brief statement of the main conclusions.

\section{Problem Statement}
\label{sec:problem_statement}

Consider an unsorted, real-valued database of size $N=2^n$, encoded as the normalized amplitudes of an $n$-qubit quantum signal
$
U_S\ket{0^n}
=
\ket{\Psi}
$
where $U_S$ is an efficient state-preparation circuit under the access model of Appendix~\ref{sec:access}. The typical amplitude scale of such a normalized state is
$
\|\ket{\Psi}\|_{\rm typ}
:=
\|\ket{\Psi}\|_2/\sqrt{N}
=
1/\sqrt{N}.
$
To characterize how strongly an extremum separates from this scale, we define the input \textit{structure factor}
\begin{equation}
\mathbb{S}_{\rm in}
:=
\frac{\|\ket{\Psi}\|_\infty}{\|\ket{\Psi}\|_{\rm typ}}
=
\sqrt{N}\,\|\ket{\Psi}\|_\infty.
\label{eq:input_structure_factor}
\end{equation}

For a physical observable sampled or computed at $N=2^n$ points as $g=(g_0,\ldots,g_{N-1})\in\mathbb R^N$, the corresponding normalized amplitudes are
\begin{equation}
u_i=\frac{g_i}{\|g\|_2}
=\frac{g_i}{\sqrt{N}\,g_{\rm rms}},
\qquad
g_{\rm rms}^2=\frac{1}{N}\sum_{i=0}^{N-1}g_i^2.
\end{equation}
It follows from \eqref{eq:input_structure_factor} that
\begin{equation}
\mathbb{S}_{\rm in}
=
\frac{g_{\rm pk}}{g_{\rm rms}},
\qquad
g_{\rm pk}:=\max_i|g_i|.
\end{equation}
{Here, $g_{\rm pk}$ and $g_{\rm rms}$ are the peak and root-mean-squared values of the input signal, respectively.} Thus, {the structure factor} $\mathbb{S}_{\rm in}$ measures the separation of the extremal amplitude from the typical signal scale.

The canonical problem is to find the largest-magnitude extremum, $i^\star\in\operatorname*{arg\,max}_{i\in[N]}|u_i|$, where $[N]:=\{0,1,\ldots,N-1\}$. The signed maximum, signed minimum, and smallest-magnitude extremum are reduced to this canonical problem under the corresponding preparation conditions described in Appendix~\ref{sec:Candidates,enum,cert,sign,cases}. {The smallest-magnitude extremum additionally requires efficient preparation by subtracting a background with the same tail properties.} Sign information for the recovered candidates is determined there by interference rather than by assuming a sign oracle. The specific objectives of this work are:

\begin{enumerate}[leftmargin=*,itemsep=1pt,topsep=2pt,parsep=0pt]

\item \textbf{Objective A} -- 
{Establish a quantum algorithm using only amplitude-state access, which returns an address $\hat{i}$ satisfying the condition
\begin{equation}
\Pr\!\left[
\hat{i}\in
\arg\max_{i\in[N]}|u_i|
\right]
\geq
1-\delta ,
\end{equation}
with overall polylogarithmic algorithm complexity in $N$; $\delta$ here is the failure probability. The returned extremal address, value, and sign are certified globally.}

\item \textbf{Objective B} -- Determine conditions 
for which the complete certified search is quantum efficient. In particular, derive the retained-mass requirement in terms of the structure factor $\mathbb{S}_{\rm in}$, together with {the required} conditions of the algorithm. Compare the resulting query complexity with classical scanning and GDH. 
\item \textbf{Objective C} -- Demonstrate the above requirements on a practical nonlinear physical problem using intermittent velocity-gradient signals from homogeneous isotropic turbulence governed by
\begin{equation}
\frac{\partial\mathbf{u}}{\partial t}+(\mathbf{u}\cdot\nabla)\mathbf{u}
=
-\frac{1}{\rho}\nabla p+\nu\nabla^2\mathbf{u}+\mathbf{f},
\qquad
\nabla\cdot\mathbf{u}=0,
\label{eq:NS}
\end{equation}
and quantify the required structure, finite-resource costs, as well as the Reynolds number dependence, and include efficient state preparation for the nonlinear PDE, extrema extraction, and classical verification. The objectives on signed maximum, signed minimum, and smallest magnitude are reduced to the canonical extrema problem under the conditions of Appendix~\ref{sec:Candidates,enum,cert,sign,cases}.
\end{enumerate}

\section{Main results}
\label{sec:main_results}

The extrema-search problem is enabled via an approximate
filtering stage that exposes the extremal candidates by suppressing the background, without attempting to resolve or compare the ordering of nearby amplitudes. This is followed by a certification stage that compares only the surviving candidates. {For a threshold $\lambda$ and a finite transition width
$\Gamma>0$, define the candidate and safely retained sets as}
{\begin{align}
&C_\Gamma
:=
\left\{
i\in[N]: |u_i|\geq\lambda
\right\},\\
&R_\Gamma
:=
\left\{
i\in[N]: |u_i|\geq\lambda+\Gamma
\right\}.
\end{align}}
{$\Gamma$ is a design parameter, and the interval between $\lambda$ and $\lambda+\Gamma$ is
intentionally left unresolved.} 
The basic operation on which the algorithm is built is this: Can a transformation be built such that an unknown amplitude $u_i$ be coherently multiplied by some function of itself, $p_d(u_i)$, without first measuring $u_i$ or writing it into a value register? The following lemma answers this question.

\begin{lemma}[Chebyshev-walk nonlinear amplitude transformation]
\label{lem:qg1}
{
Let
$p_d(x)=\sum_{\ell=0}^{d}c_\ell T_\ell(x)$
be a real degree-$d$ polynomial, bounded on $[-1,1]$, where
$T_\ell$ is the Chebyshev polynomial of the first kind. From
$U_S$, $U_S^\dagger$, and their controlled versions, one can
construct the amplitude observable
$D_\Psi:=\sum_i u_i|i\rangle\langle i|$ and a Chebyshev-walk
circuit whose successful branch containing the solution is
\begin{equation}
\frac{1}{C_d}p_d(D_\Psi)|\Psi\rangle
=
\frac{1}{C_d}\sum_i p_d(u_i)u_i|i\rangle,
\qquad
C_d:=\sum_{\ell=0}^{d}|c_\ell|.
\end{equation}
If
$P_e:=\|p_d(D_\Psi)|\Psi\rangle\|_2^2$,
the successful branch occurs with probability $P_e/C_d^2$ and can be
prepared, using amplitude amplification, with query complexity
\begin{equation}
Q_{\rm QG1}^{\rm AA}
=
O\!\left(
\frac{dC_d}{\sqrt{P_e}}
\log\frac{1}{\delta_{\rm post}}
\right),
\end{equation} where $\delta_{\rm post}$ is the allowed preparation failure
probability.}
\end{lemma}
The construction and proof are given in
Appendices~\ref{sec:chebyshev_walk}--\ref{sec:theorem_proof}.

Choosing $p_d$ close to 0 below the threshold and close
to 1 on the safely-retained tail transforms
$
\sum_i u_i|i\rangle
\quad\longmapsto\quad
\sum_i p_d(u_i)u_i|i\rangle ,
$
so that the background is suppressed while the extreme amplitudes remain. We call this operation \emph{Quantum Gadget 1} (QG1).  Its circuit construction is given in Sec.~\ref{sec:qg1}, while the existence and the degree of the required bounded polynomial are proved in Appendix~\ref{sec:polynomial_filter}, and the Chebyshev-walk realization of this polynomial is proved in Appendix~\ref{sec:chebyshev_walk}.

Filtering alone is not enough. If too many amplitudes survive, the search space has not been reduced appreciably; if the retained tail carries too little probability, preparing it becomes expensive. Efficient extrema search therefore requires two properties: only a small number of candidate addresses should remain, and those candidates must collectively carry sufficient normalized probability mass. These are the structural conditions entering the main result of Theorem 1, which is stated following the two explanations.

Let $K:=|C_\Gamma|$ be the number of candidate addresses and
$ P_R:=\sum_{i\in R_\Gamma}|u_i|^2>0 $
be the probability carried by the safely retained tail.Further, denote by $p_d$ the bounded Chebyshev filter of Lemma~1, with
coefficient norm $C_d:=\sum_{\ell=0}^{d}|c_\ell|$, and error
$\eta_{\rm filt}\in(0,1/2)$, such that
$|p_d(x)|\leq\eta_{\rm filt}$ for $|x|\leq\lambda$ and
$|p_d(x)|\geq1-\eta_{\rm filt}$ for
$|x|\geq\lambda+\Gamma$.

\begin{theorem}[Certified extrema search for structured quantum amplitude signals]
\label{thm:main}

Supposing that the filter leakage is controlled relative to the
retained mass, $\eta_{\rm filt}^2\leq\xi P_R,~ 0<\xi<1,$ where $\xi$ is independent of $N$. Then, for any target failure
probability $\delta\in(0,1)$, a global largest-magnitude extremum
can be recovered and certified using
\begin{equation}
Q_{\rm ext}^{\rm AA}
=
\widetilde O\!\left(
\frac{K d C_d}{\sqrt{P_R}}
+
\frac{K}{\epsilon_{\rm QAE}}
\right),
\qquad
d=
O\!\left(
\Gamma^{-1}\log\frac{1}{\eta_{\rm filt}}
\right),
\end{equation}
coherent state-preparation queries. Here
$\epsilon_{\rm QAE}$ is the additive precision used in the final
quantum-amplitude-estimation step.

{For a unique extremum, define $p_i:=|u_i|^2$ and let
$p_{(1)}>p_{(2)}$ be the largest and second-largest candidate
probabilities. The choices
$
\epsilon_{\rm QAE}
<
\Delta_p/4,$ and $
\Delta_p:=p_{(1)}-p_{(2)},
$
are sufficient to distinguish the global extremum from the
second-largest candidate.}
\end{theorem}

{For the representative regime studied below, let
$L:=\log_2N$, with $\Gamma_N=L^{-3/2}$,
$K=O(L)$, and $P_R=\Omega(L^{-1})$. The complexity in terms of the one-time threshold
calibration cost, $Q_{\rm cal}$, is given by
\begin{equation}
Q_{\rm ext}^{\rm AA}
=
Q_{\rm cal}
+
\widetilde O\!\left(
C_dL^3+\frac{L}{\epsilon_{\rm QAE}}
\right).
\end{equation}
The more general inverse-polylogarithmic result, together with the
exact leakage and confidence factors, is given in
Appendices~\ref{sec:polynomial_filter}--\ref{sec:theorem_proof} and~\ref{app:global_completeness}.}

The QG1 construction and candidate recovery are described in Secs.~\ref{sec:qg1} and~\ref{sec:qg2}. The complete proof of Theorem~\ref{thm:main}, including the polynomial bounds, recovery and certification costs is given in Appendices~\ref{sec:polynomial_filter}--\ref{sec:theorem_proof} and~\ref{sec:Candidates,enum,cert,sign,cases}--\ref{app:global_completeness}. 

Thus, under the structure conditions specied above, the certified extrema problem remains polylogarithmic in $N$ under the stated accuracy and failure assumptions, compared with the $O(\sqrt{N})$ GDH and $O(N)$ classical database size dependencies under their corresponding access models.

As pointed out in a recent work \cite{yamakawa2024verifiable}, a useful, end-to-end quantum advantage additionally requires that candidate compression preserve the true global extremum, that the result be efficiently classically verifiable, and that the input state itself be generated without an $O(N)$ classical loading step. We establish these properties next.

\begin{figure*}[t]
\centering
\includegraphics[trim={0.1cm 0.1cm 0.1cm 0.1cm},clip=true,scale=0.27]{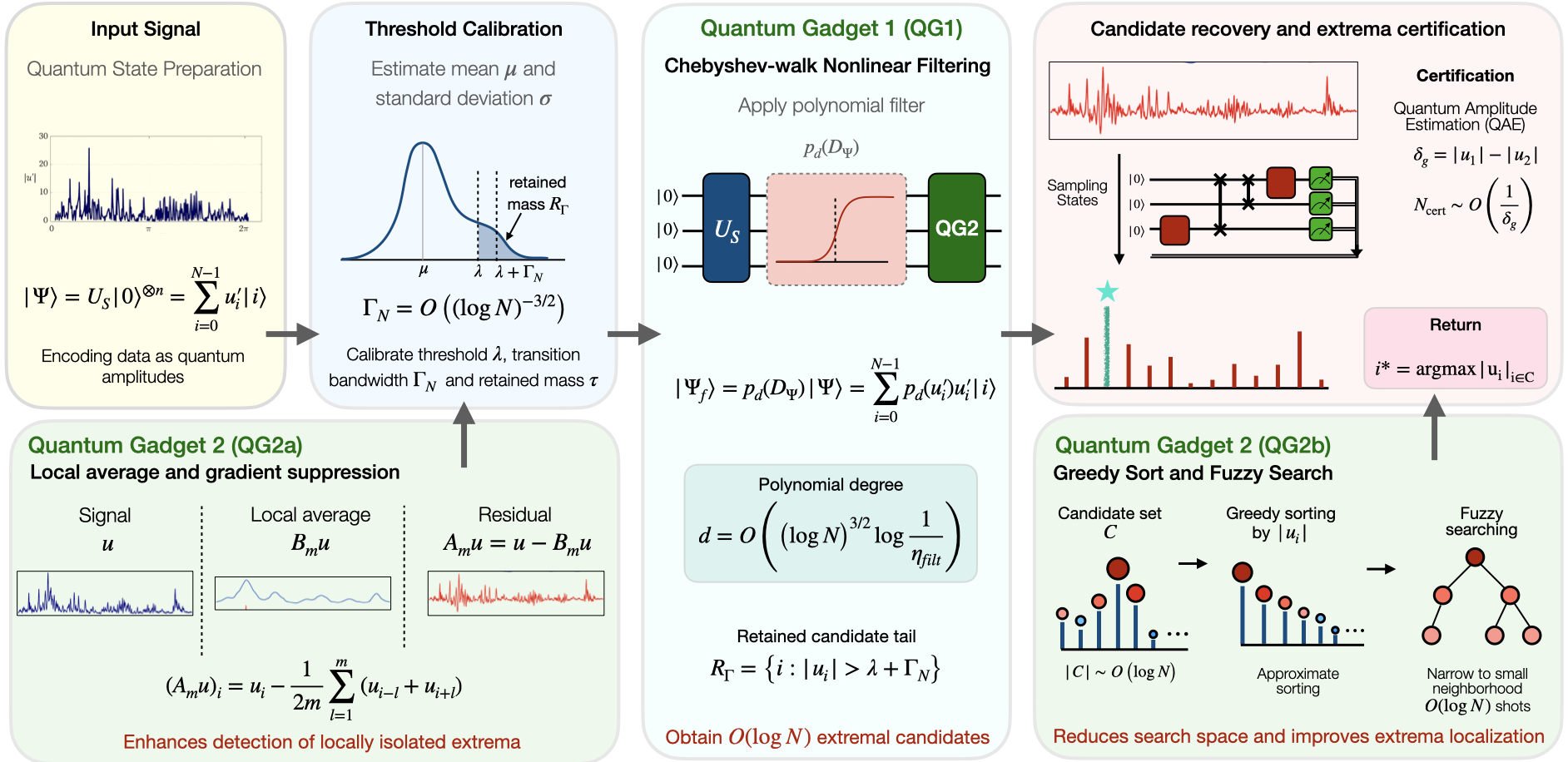}
\caption{\justifying \textbf{End-to-end workflow of the quantum extrema-search algorithm} -- The input amplitude state is first calibrated to determine the threshold $\lambda$ and finite \textcolor{blue}{filter threshold} width $\Gamma_N$. QG1 then applies the Chebyshev-walk polynomial filter, suppressing the background while retaining the extreme-event tail. Repeated successful measurements recover $O(\log N)$ candidate addresses, after which only these candidates are compared and the extremal address, along with its value and sign, is certified. QG2a optionally enhances extrema that stand out from their local surroundings by subtracting a neighborhood average before filtering, while QG2b optionally localizes a dominant candidate using the binary address structure after filtering. The unmodified QG1 route and explicit candidate enumeration remain rigorous fallbacks.}
\label{fig:workflow_main}
\end{figure*}
\begin{lemma}[Global completeness after candidate compression]
\label{lem:global_completeness}
{Let
$i^\star\in\arg\max_i |u_i|$
be a global extremum. Under the conditions of
Theorem~1, $P_R>0$ implies that the global extremum lies in the
safely retained set $R_\Gamma$. In terms of $P_e
:=
\sum_i |p_d(u_i)|^2|u_i|^2$ from Lemma 1, write
\begin{equation}
q_\star
:=
\frac{|p_d(u_{i^\star})|^2|u_{i^\star}|^2}{P_e},
\end{equation}
to denote the probability of observing $i^\star$ in one successful
filtered-state measurement. Then
\begin{equation}
q_\star
\geq
\frac{(1-\eta_{\rm filt})^2}
{K\left(1+\eta_{\rm filt}^2/P_R\right)} .
\end{equation}
Consequently, it is sufficient to use 
\begin{equation}
M_{\rm glob}
=
O\!\left[
\frac{K}{(1-\eta_{\rm filt})^2}
\left(
1+\frac{\eta_{\rm filt}^2}{P_R}
\right)
\log\frac{1}{\delta_{\rm glob}}
\right]
\end{equation}
number of successful filtered-state measurements to recover
a global-extremum address with probability at least
$1-\delta_{\rm glob}$, where $\delta_{\rm glob}$ is the allowed
failure probability of the candidate-recovery stage.}
\end{lemma}
{In particular, when the filter leakage is controlled so that
$\eta_{\rm filt}^2=O(P_R)$, the required number of successful
measurements is
$M_{\rm glob}=O[K\log(1/\delta_{\rm glob})]$.
Thus the compression step retains a true global extremum of the
original $N$-point signal with controlled failure probability,
rather than merely returning the largest event among the sampled
candidates. The proof is given in Appendix~\ref{app:global_completeness}.}

Global completeness establishes the correctness of the extrema search algorithm. We next ask whether that compact result can also be verified efficiently by a classical observer. We distinguish verification under faithful circuit execution from verification when the quantum device itself is not trusted (``untrusted"). 

\begin{proposition}[Classical verification \vspace{0.3cm}]
\label{prop:verification}

\textbf{\newline Level I (trusted device verification)}\\

Let
\begin{equation}
q_i:=
\frac{|p_d(u_i)|^2|u_i|^2}{P_e}
\end{equation}

\noindent be the probability of observing address $i$ in the successful
filtered state. If $\widehat C$ is the recovered candidate set,
define the unrecovered filtered probability

\begin{equation}
R_{\rm miss}:=
1-\sum_{i\in\widehat C}q_i .
\end{equation}
Let $\hat i\in\widehat C$ be the recovered candidate with the
largest certified amplitude. If $\hat i\in R_\Gamma$, then
\begin{equation}
R_{\rm miss}
<
(1-\eta_{\rm filt})^2q_{\hat i}
\end{equation}

\noindent certifies that no unrecovered address has amplitude larger than
$|u_{\hat i}|$. Once these estimates have been obtained, 
the remaining verification is classical.\\ 

\noindent \textbf{Level II (untrusted quantum device)}
The complete extrema-search computation can be reduced to a
bounded-error quantum decision circuit and verified using the
classical cryptographic protocol of ref.~\cite{mahadev2018classical}, under its stated
computational assumptions. The verifier cost is polynomial in the
verified circuit size and security parameter, and is therefore
polylogarithmic in $N$ when the latter two are also polylogarithmic in $N$.

\end{proposition}

Level I verifies the reported extremum directly from the compressed filtered state under faithful circuit execution. Level II removes this trust assumption by reducing the claimed extremum to a bounded-error quantum decision problem that can be verified classically using the protocol of ref.~\cite{mahadev2018classical}. The Level-I certificate and the Level-II reduction are derived in
Appendix~\ref{app:verification}.

For fluid dynamics applications, the state-preparation module $U_S$ can be implemented by the quantum homotopy solver of ref.~\cite{bharadwaj2025homotopy}, followed when required by the derived-field operation of ref.~\cite{bharadwaj2024compact}. Let $Q_S:=Q_H+Q_{\nabla}$ denote the query cost of one coherent preparation, including its inverse and controlled use. Under the stated assumptions of these algorithms, $Q_S=\polylog{N}$, and the corresponding underlying gate complexity is also polylogarithmic in the resolved problem size. Thus, substituting this preparation into the extrema-search algorithm gives the end-to-end query complexity of
\begin{equation}
Q_{\rm e2e}
=
Q_S Q_{\rm ext}^{\rm AA}
=
\polylog{N}.
\label{eq:end_to_end_query}
\end{equation}
The explicit composition, including sparsity, conditioning, postselection, accuracy, and certification factors, is given in Appendix~\ref{sec:end_to_end}. The complete workflow is summarized in Fig.~\ref{fig:workflow_main}.

\section{How the algorithm works}
\label{sec:how_algorithm}

The algorithm separates the difficult global search from the final high-precision comparison, as summarized in Fig.~\ref{fig:workflow_main}. Starting from the amplitude state $|\Psi\rangle=\sum_i u_i|i\rangle$, coarse global statistics are first used only to place the filter threshold $\lambda$ and its finite transition width $\Gamma_N$. QG1 then applies a bounded Chebyshev polynomial directly to the amplitudes, suppressing the broad background while retaining the extreme-event tail. Measurements of the successful filtered state therefore return candidates rather than arbitrary indices, reducing the original $N$-point problem to $O(\log N)$ relevant addresses. QG2 provides two optional refinements. Before QG1, a local-excess transformation can enhance peaks or troughs relative to their neighborhood; after QG1, fuzzy search and greedy sorting can localize a dominant candidate through the binary address tree. Finally, only the recovered candidates are compared at high precision to determine the extremal value and sign, while the global-completeness and verification tests certify that a better unrecovered candidate has not been missed. Thus the expensive global stage performs only coarse candidate isolation, and the fine numerical comparison is postponed until the search space has already been compressed.

\subsection{Quantum Gadget 1: Chebyshev-walk amplitude filtering}
\label{sec:qg1}

The central task of QG1 is to transform the input amplitude state $|\Psi\rangle=\sum_i u_i|i\rangle$ into a state in which the broad background is suppressed and the extreme amplitudes are retained, without measuring or explicitly estimating the individual values $u_i$. This corresponds to the nonlinear amplitude transformation
\begin{equation}
\sum_i u_i|i\rangle
\longmapsto
\sum_i p_d(u_i)u_i|i\rangle,
\label{eq:qg1_goal}
\end{equation}
where the bounded polynomial $p_d$ is chosen so that $p_d(u_i)\simeq0$ below the candidate threshold and $p_d(u_i)\simeq1$ in the safely retained tail. The key question is therefore how to apply a scalar polynomial directly to amplitudes that are available only implicitly through the preparation circuit $U_S$.

First, the amplitudes must be made accessible to quantum operations without being read out. Let $S$ be the signal register and $A$ an auxiliary register. A computational-basis copy followed by $U_S^\dagger$ defines
\begin{align}
&R_\Psi=(I_S\otimes U_S^\dagger)C,~~~
(I_S\otimes\langle0^n|_A)R_\Psi(I_S\otimes|0^n\rangle_A)
=
D_\Psi,\\
&D_\Psi:=\sum_i u_i|i\rangle\langle i|.
\label{eq:qg1_block}
\end{align}
Thus the unknown amplitude $u_i$ has become the eigenvalue of $D_\Psi$ associated with $|i\rangle$, since $D_\Psi|i\rangle=u_i|i\rangle$. No numerical value register or amplitude measurement is introduced. The construction and its proof are given in Appendix~\ref{sec:chebyshev_walk}.

Second, we need to apply a useful function to these eigenvalues. This is where the Chebyshev walk enters, which we notice can be  related to the qubitization and block-encoded polynomial-transformation framework~\cite{LowChuang2019,Gilyen2019}, specialized here to the amplitude observable $D_\Psi$ constructed directly from $U_S$. After adding one qubit and converting $R_\Psi$ into the Hermitian unitary $\mathcal U_\Psi$, define
\begin{equation}
\Pi
=
|+\rangle\langle+|_q
\otimes
|0^n\rangle\langle0^n|_A
\otimes I_S,
\qquad
W_\Psi=(2\Pi-I)\mathcal U_\Psi .
\label{eq:qg1_walk}
\end{equation}
The key identity is
\begin{equation}
\Pi W_\Psi^\ell\Pi
=
T_\ell(D_\Psi)\Pi .
\label{eq:qg1_cheb_identity}
\end{equation}
Its meaning is simple. For each signal amplitude write $u_i=\cos\theta_i$. In the two-dimensional subspace associated with index $i$, one application of $W_\Psi$ is a rotation by the angle $\theta_i$. After $\ell$ applications the angle becomes $\ell\theta_i$, and projection back onto the selected auxiliary state gives
$\cos(\ell\theta_i)=T_\ell(u_i)$. The walk therefore converts the hidden amplitudes into the complete Chebyshev basis
$1,u_i,T_2(u_i),\ldots,T_d(u_i)$ without learning any $u_i$ individually. The full rotation proof is given in Appendix~\ref{sec:chebyshev_walk}.

Finally, the required filter is already written in this basis,
\begin{equation}
p_d(x)=\sum_{\ell=0}^{d}c_\ell T_\ell(x),
\qquad
C_d:=\sum_{\ell=0}^{d}|c_\ell|.
\label{eq:qg1_polynomial}
\end{equation}
As shown schematically in Fig.~\ref{fig:qg1_circuit_main}, a history register prepares the different degrees $\ell$ with weights proportional to $\sqrt{|c_\ell|}$, controls the corresponding powers $W_\Psi^\ell$, and inserts the sign of $c_\ell$. Unpreparing the history register contributes the second factor $\sqrt{|c_\ell|}$, so the successful branch coherently adds the terms with their required coefficients. Projecting the auxiliary registers onto their selected states therefore gives
\begin{align}
|\Xi_{\rm succ}\rangle
&=
\frac{1}{C_d}
\sum_{\ell=0}^{d}c_\ell T_\ell(D_\Psi)|\Psi\rangle \\
&=
\frac{1}{C_d}p_d(D_\Psi)|\Psi\rangle
=
\frac{1}{C_d}\sum_i p_d(u_i)u_i|i\rangle .
\label{eq:qg1_output}
\end{align}
The coefficient synthesis and normalization are derived in Appendix~\ref{sec:theorem_proof}.

\eqref{eq:qg1_output} is the central action of QG1. Amplitudes below the threshold are multiplied by a number close to zero and are therefore suppressed, while amplitudes safely beyond the transition region are multiplied by a number close to one and retain approximately their original weight. Importantly, the filter does not determine which surviving amplitude is the largest. It only concentrates the quantum state onto the short candidate tail, leaving the fine comparison to the later certification stage.

If $P_e:=\|p_d(D_\Psi)|\Psi\rangle\|_2^2$, the branch in \eqref{eq:qg1_output} occurs with probability $P_e/C_d^2$. Amplitude amplification gives
\begin{equation}
Q_{\rm QG1}^{\rm AA}
=
O\!\left(
\frac{dC_d}{\sqrt{P_e}}
\log\frac{1}{\delta_{\rm post}}
\right),
\label{eq:qg1_cost}
\end{equation}
where $\delta_{\rm post}$ is the allowed preparation failure probability.
{For retained mass $P_R$, the filter guarantee gives
$P_e\geq(1-\eta_{\rm filt})^2P_R$. Thus the preparation
cost is controlled directly by the retained probability mass
rather than by any prescribed logarithmic form.}

{The remaining ingredient is the degree $d$. QG1 deliberately
does not ask the polynomial to resolve the potentially tiny
separation between the two largest amplitudes. It resolves only
the designed transition width $\Gamma$. The bounded polynomial
construction in Appendix~\ref{sec:polynomial_filter} gives
$
d
=
O\!\left(
\Gamma^{-1}
\log\frac{1}{\eta_{\rm filt}}
\right).
$
For the particular case of turbulence
$\Gamma=\Gamma_N=(\log_2N)^{-3/2}$, so that $d$ becomes
$d=O[(\log_2N)^{3/2}\log(1/\eta_{\rm filt})]$.
Thus the expensive global stage depends on the deliberately
chosen finite transition width rather than on the final top-two
spacing. That fine spacing enters only after QG1 has compressed
the original signal to the candidate set.}

\begin{figure*}[htpb!]
\centering
\includegraphics[trim={0.1cm 0.1cm 0.1cm 0.1cm},clip=true,scale=0.25]{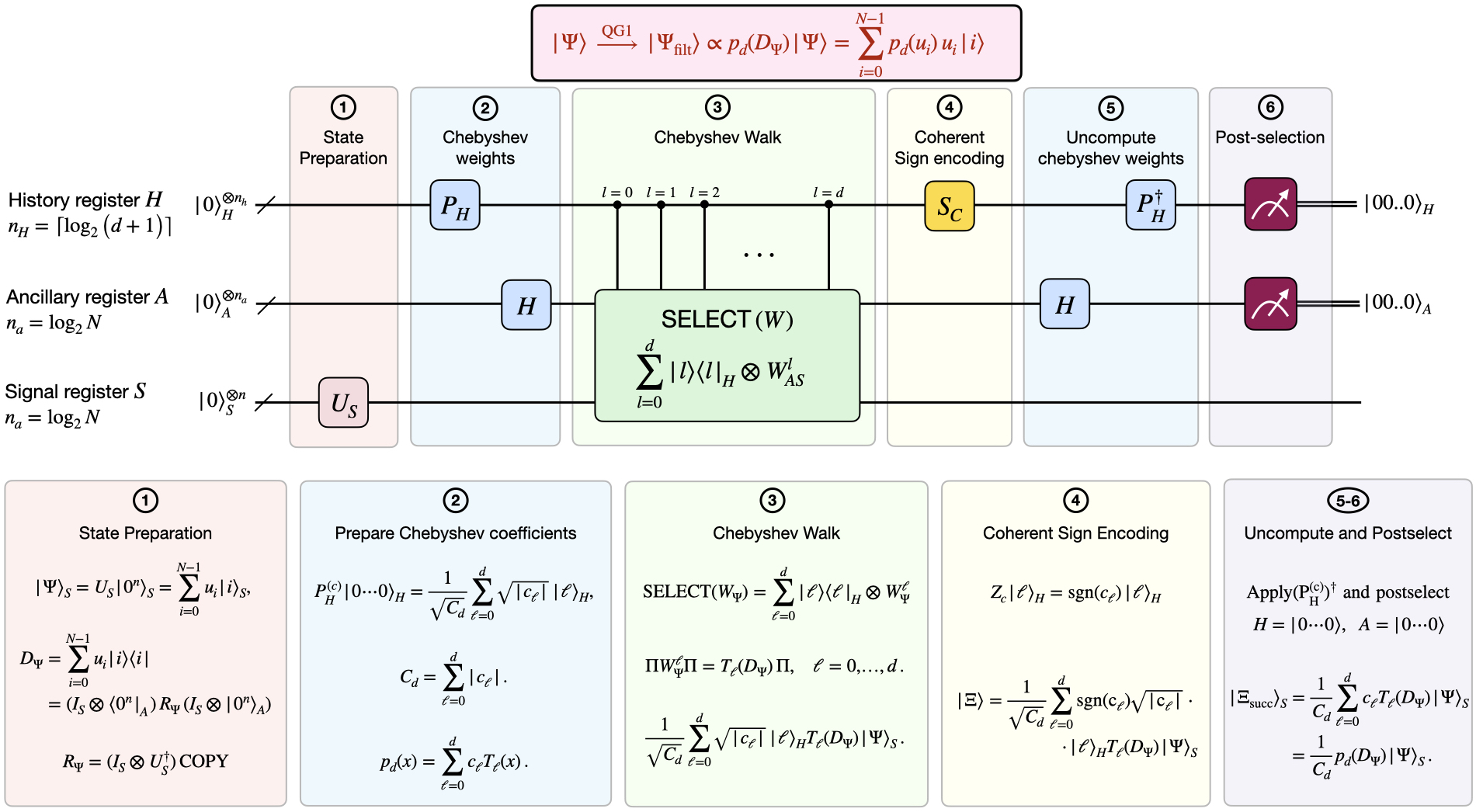}
\caption{\justifying \textbf{Register-level construction of QG1.} The signal register is prepared as $|\Psi\rangle=U_S|0^n\rangle$. The amplitude-observable construction exposes the unknown amplitudes through $D_\Psi$, and the Chebyshev walk generates the operators $T_\ell(D_\Psi)$. The history register $H$ is prepared with amplitudes $\sqrt{|c_\ell|/C_d}$, controls the corresponding walk powers $W_\Psi^\ell$, and acquires the sign of each coefficient $c_\ell$. Unpreparing and postselecting the auxiliary registers coherently combines these terms, leaving $C_d^{-1}p_d(D_\Psi)|\Psi\rangle$ in the signal register. The lower row summarizes the corresponding algebraic state evolution.}
\label{fig:qg1_circuit_main}
\end{figure*}

\subsection{Quantum Gadget 2: optional contrast and localization}
\label{sec:qg2}

QG2 provides two optional refinements to QG1. QG2a acts before filtering and enhances events that are large relative to their local surroundings. QG2b acts after filtering and uses the binary structure of the address register to localize a dominant candidate. Neither operation is required for Theorem~\ref{thm:main}, and the original QG1 route remains the rigorous fallback.

For QG2a, let $m$ neighboring points on either side of $u_i$ define the local average
\begin{equation}
(B_m u)_i
=
\frac{1}{2m}
\sum_{\ell=1}^{m}
\left(
u_{i-\ell}+u_{i+\ell}
\right),
\qquad
A_m:=I-B_m ,
\label{eq:qg2_local}
\end{equation}
with periodic indexing when appropriate. The residual $(A_m u)_i=u_i-(B_m u)_i$ is small when the signal varies slowly near $i$ but remains large for a sufficiently isolated peak or trough. Thus $m$ controls the spatial scale over which local contrast is measured. Importantly, $A_m$ does not in general preserve the ordering of the original amplitudes. QG2a is therefore used only when extrema relative to the local surroundings are physically relevant; the canonical global-extremum problem corresponds to $A_0:=I$.

Writing $S|i\rangle=|i+1\pmod N\rangle$ gives $B_m=(2m)^{-1}\sum_{\ell=1}^{m}(S^\ell+S^{-\ell})$, whose LCU implementation has coefficient normalization $C_{A,m}=2$ for $m\geq1$. For cost accounting, define
\begin{align}
&E_m:=\|A_m|\Psi\rangle\|_2^2,
\qquad
~~~~~~|\Psi_m\rangle:=\frac{A_m|\Psi\rangle}{\sqrt{E_m}},
\\
&F_m:=\|p_{d_m}(D_m)|\Psi_m\rangle\|_2^2,
~~~P_m:=E_mF_m ,
\label{eq:qg2_useful_mass}
\end{align}
where $D_m$ is the amplitude observable of $|\Psi_m\rangle$. The product $P_m$ accounts both for preparing the residual and for retaining its extreme-event tail after QG1, preventing normalization of a very small residual from appearing artificially favorable.

\begin{proposition}[Cost of the optional local contrast transformation]
\label{prop:qg2_channel}
For any admissible neighborhood width $m$, one successful QG1 candidate state prepared after applying $A_m$ has query cost
\begin{align}
&Q_m^{\rm AA}
=
O\!\left[
\frac{C_{A,m}d_mC_{d_m}}{\sqrt{P_m}}
\log\!\left(
\frac{1}{\delta_m}
\right)
\right],\\
&d_m
=
O\!\left[
(\log_2N)^{3/2}
\log\!\left(
\frac{1}{\eta_{{\rm filt},m}}
\right)
\right].
\label{eq:qg2_channel_cost}
\end{align}
A prescribed set $\mathcal M$ of widths can be tested during calibration, always including $m=0$. The complete implementation and cost accounting are given in Appendix~\ref{app:qg2}.
\end{proposition}

QG2b instead uses the binary structure of the index register to localize a dominant candidate after filtering. For $N=2^L$, each candidate address contains $L=\log_2N$ bits. The fuzzy search successively asks which of the two current binary blocks contains more filtered probability, without estimating the amplitudes inside those blocks. The accompanying greedy sort records the statistically favored orientation at each level and temporarily moves that branch to a fixed location. Figure~\ref{fig:qg2_main} illustrates this procedure.

Let $B_r$ be the depth-$r$ block on the path containing the true extremum $i^\star$. Define
\begin{align}
&P_r:=\Pr[i\in B_r]~~
\textrm{and}\\
&\Pr[
\text{child containing }i^\star
\mid i\in B_r
]
=
\frac{1+\zeta_r}{2}.
\label{eq:qg2_hierarchy_def}
\end{align}

\begin{proposition}[Fuzzy search and greedy sort under a hierarchy promise]
\label{prop:qg2_localization}
Suppose
\begin{equation}
P_r\geq p_0>0,
\qquad
\zeta_r\geq\zeta_0>0,
\qquad
r=0,\ldots,L-1.
\label{eq:qg2_hierarchy}
\end{equation}
Then the complete $L$-bit address is recovered with failure probability at most $\delta_{\rm fg}$ using
\begin{equation}
M_{\rm fg}
=
O\left\lceil
\frac{1}{p_0\zeta_0^2}
\log\!\left(
\frac{2L}{\delta_{\rm fg}}
\right)
\right\rceil.
\label{eq:qg2_samples}
\end{equation}
successful candidate-state measurements.
\end{proposition}

A measured candidate address already contains all $L$ bits, so the same collection of measurements can be reused at every binary level rather than preparing a fresh batch for each level. If $G_r$ denotes the reversible swap used to place the favored segment at a fixed branch, the accumulated greedy sorting is $G_{\rm gr}=G_{L-1}\cdots G_0$. Once the terminal candidate is found at address $j^\star$, its original physical index is recovered simply by preparing $|j^\star\rangle$ and undoing the recorded sorting,
\begin{equation}
G_{\rm gr}^{\dagger}|j^\star\rangle
=
|i^\star\rangle .
\label{eq:qg2_retracing}
\end{equation}
QG2b should be regarded as an optional localization procedure rather than an independent asymptotic speedup. Its statistical proof and retracing construction are given in Appendix~\ref{app:qg2}; when the hierarchy condition is unavailable, explicit candidate enumeration described in Sec.~\ref{subsec:enuemration fallback} remains the rigorous fallback.
\begin{figure}[!t]
\centering
\includegraphics[width=1.0\linewidth]{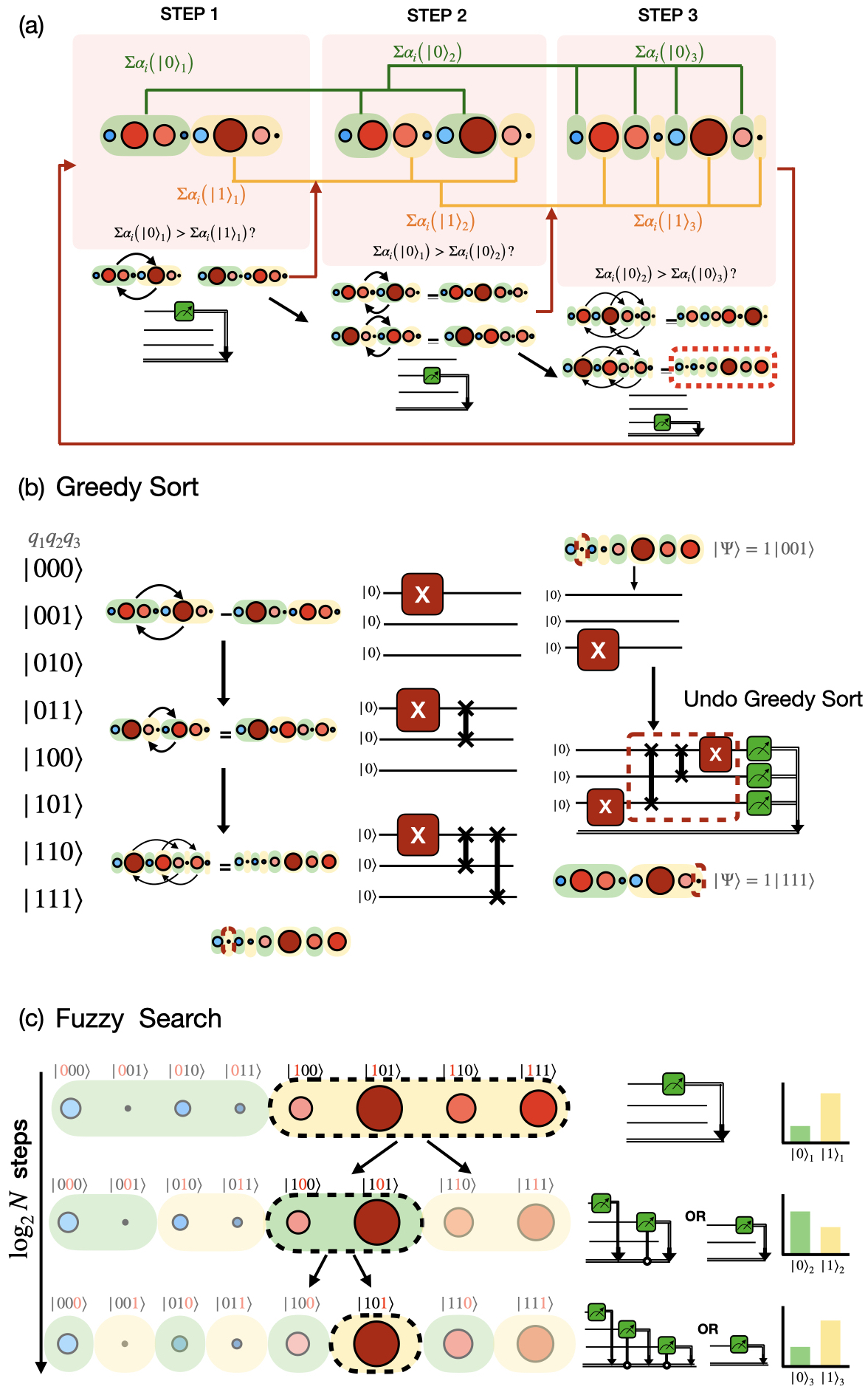}
\caption {\justifying \textbf{QG2b hierarchical candidate localization.} (a) Candidate probability is compared recursively across binary prefix blocks. At each level, the more strongly weighted child is recorded when the evidence is sufficient, progressively identifying the prefix containing the dominant candidate. (b) Controlled swaps implement the recorded reorientations, while the inverse sequence restores the original address ordering after localization. (c) The hierarchical search follows the statistically favored child at successive binary levels and retains both branches when the evidence is inconclusive. Under the path-bias condition of Proposition~\ref{prop:qg2_localization}, this procedure recovers the full $\log_2N$-bit address with only logarithmic confidence overhead.}
\label{fig:qg2_main}
\end{figure}

\begin{figure*}
\centering
\includegraphics[trim={0.2cm 0.1cm 0.1cm 0.1cm},clip=true,scale=0.29]{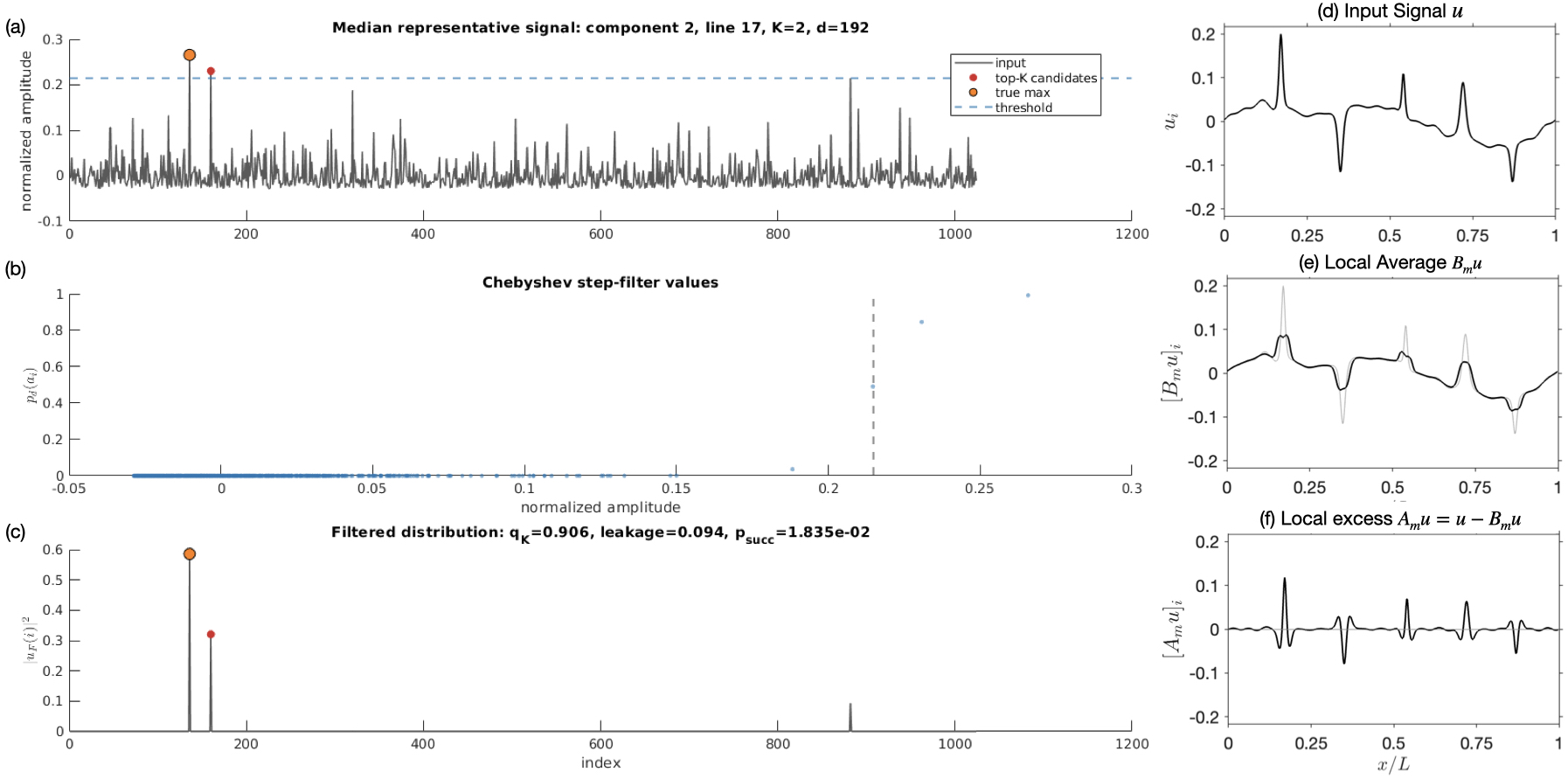}
\caption{ \justifying \textbf{Representative extrema filtering and local-background subtraction.}
(a) Normalized $N=1024$ turbulent velocity-gradient signal. The dashed line marks the calibrated candidate threshold, red markers denote the retained leading candidates, and the orange marker identifies the true maximum. 
(b) Degree-$192$ bounded Chebyshev multiplier used by QG1. The broad background is strongly suppressed, while amplitudes beyond the finite transition region are retained. 
(c) Probability distribution conditioned on successful filtering. The leading candidate set carries $q_K=0.906$ of the output probability and the true maximum becomes the dominant sampling outcome, illustrating the reduction from a full $N$-point search to a short candidate problem. 
(d)--(f) Illustration of the optional QG2a transformation on a representative signal segment. The input $u$ is compared with the average $B_m u$ of its neighboring values, here shown for $m=10$, and their difference $A_m u=u-B_m u$ removes the slowly varying local background while retaining features that are large relative to their surroundings.}
\label{fig:qg1_representative_main}
\end{figure*}

\section{Numerical validation on turbulent velocity gradients}
\label{sec:numerics}

{The quantum advantage considered in Sec.~\ref{sec:main_results} applies when the input amplitude state is generated quantumly rather than classically. The present algorithm should therefore be interpreted as
an independent validation of the search algorithm, not as a claim that the
current algorithm provides a PDE solver that generates the required data in a quantum format. Additional computing will be necessary if the data are provided in classical representation.} {In fact, recent proposals ~\cite{zhao2026exponential} point in the direction of providing a QRAM-free route from streamed classical data to access approximate coherent state-preparation, with the associated sampling cost accounted for separately.} 

{For the case of turbulent velocity gradients considered here, each classical data point $g=(g_0,\ldots,g_{N-1})$ is first prepared through
$U_g{|0\rangle}^{\otimes n}=|\Psi_g\rangle=\sum_{i=0}^{N-1}(g_i/\|g\|_2)|i\rangle$. Because the values $g_i$ are supplied from a classical database, this brute-force state-preparation step can require $O(N)$ operations which, if included in the counting, does not yield an end-to-end quantum advantage.

Even in this setting, the algorithm demonstrates how a demanding physical test bed for detecting extreme events in the presence of strong intermittency and multiscale variability. It can be thought to stress-test the quantum extrema-search and determine whether the structural conditions, candidate compression, and finite-degree filtering mechanism required by Theorem~\ref{thm:main} continue to hold for realistic quantities of practical interest.}

The JHTDB flow is the result of solving forced Navier--Stokes equations by direct numerical simulation in a $1024^3$ box, generating homogeneous isotropic turbulence in a $(2\pi)^3$ domain~\cite{li2008public,yeung2012dissipation}. From one snapshot we extract the plane $y=\pi$, choose 64 equally spaced $z$ locations, and retain all $N=1024$ points along each periodic $x$ line. Each line $g=(g_0,\ldots,g_{N-1})$ is normalized as $u_i=g_i/\|g\|_2$. Summing up all nine components $\partial u_j/\partial x_k$ of the gradient tensor gives $64\times9=576$ line cuts of data. For $N=1024$, the designed target mass is $P_{\rm tar}=1/\log_2N=0.1$ and the transition width $\Gamma_N=(\log_2N)^{-3/2}=0.031623$. A Jackson-damped Chebyshev filter is evaluated over a degree scan, with $d_{\min}$ denoting the smallest tested degree satisfying the retained-mass and concentration conditions. The numerical diagnostics are defined in Appendix~\ref{app:numerics}.

Figure~\ref{fig:qg1_degree_main} {shows the corresponding behavior across the 576-signal ensemble.} The median leading-candidate mass exceeds the $1/\log_2N$ target, while QG1 concentrates $0.89282$ of the filtered probability on a median of only five candidates and increases the probability of observing the true maximum by a median factor of $7.8677$. Table~\ref{tab:fixedN} summarizes the remaining finite-$N$ diagnostics and the matched QAE-based resource comparison. These results test the retained-mass, candidate-compression, and finite-degree mechanisms of Theorem~\ref{thm:main}; their asymptotic scaling follows from the analytical bounds. The resource comparison uses matched amplitude-only value estimation and is not a comparison with GDH supplied with an exact (cost free) comparator.

Figure~\ref{fig:qg1_representative_main} shows the mechanism on one representative signal. The original line contains a broad intermittent background surrounding a few intense events. The bounded polynomial suppresses this background without resolving the ordering of amplitudes inside the transition region. After successful filtering, most of the probability is concentrated on only a few locations and the probability of directly sampling the true maximum is strongly increased.To make the action of QG2a transparent, Fig.~\ref{fig:qg1_representative_main}d--f uses a physically motivated synthetic model rather than a DNS realization. We take $g(x)=b(x)+\sum_{k=1}^{K_e}
A_k\operatorname{sech}^2[(x-x_k)/\delta_k]+\eta(x)$, where $b(x)$ is a slowly varying background, the localized $\operatorname{sech}^2$ terms model both positive and negative velocity-gradient events associated with thin shear layers, and $\eta(x)$ is a weak band-limited fluctuation. After normalization, $u_i=g(x_i)/\|g\|_2$, the local average $B_m u$ follows the slowly varying component, while $A_m u=u-B_m u$ suppresses it and retains events that are large relative to their neighborhood.

\begin{table}[!t]
\centering
\caption{\textbf{Fixed-$N$ numerical summary for the 576 turbulent signals at $N=1024$.}
Values are ensemble medians unless otherwise stated. The resource comparison uses the same QAE-based value-certification model for the present method and GDH.}
\label{tab:fixedN}
\footnotesize
\setlength{\tabcolsep}{4pt}
\begin{tabular}{l c l c}
\toprule
Quantity & Value & Quantity & Value \\
\midrule
$P_K$ & 0.11464
&
$q_K$ & 0.89282
\\

$K_{\rm fuzzy}$ & 5 $(36\ {\rm max})$
&
$P_{\max}^{\rm out}/P_{\max}^{\rm in}$ & 7.8677
\\

$d_{\min}$ & 192
&
$C_d$ & 2.5559
\\

$Q_{\rm present}$ & 2063.6
&
$Q_{\rm GDH+QAE}/Q_{\rm present}$ & 1.5201
\\
\bottomrule
\vspace{-1cm}
\end{tabular}
\end{table}

The corresponding finite-resource comparison is summarized in Table~\ref{tab:fixedN}. Using the same QAE-based value-estimation model for the present candidate list and for a GDH comparator, the median estimated costs are $2063.6$ and $2693.2$, respectively, giving a median signal-wise ratio of $1.5201$. This does not include the repeated-state factor required by Lemma 2. This comparison isolates the benefit of compressing amplitude-only data before numerical value estimation. It should not be interpreted as an advantage over an ideal GDH algorithm supplied with a cost-free comparison oracle. The asymptotic advantage instead follows from Theorem~\ref{thm:main} together with the physical structure condition discussed below.

\begin{figure*}
\centering
\includegraphics[trim={0.2cm 0.2cm 0.2cm 0.2cm},clip=true,scale=0.27]{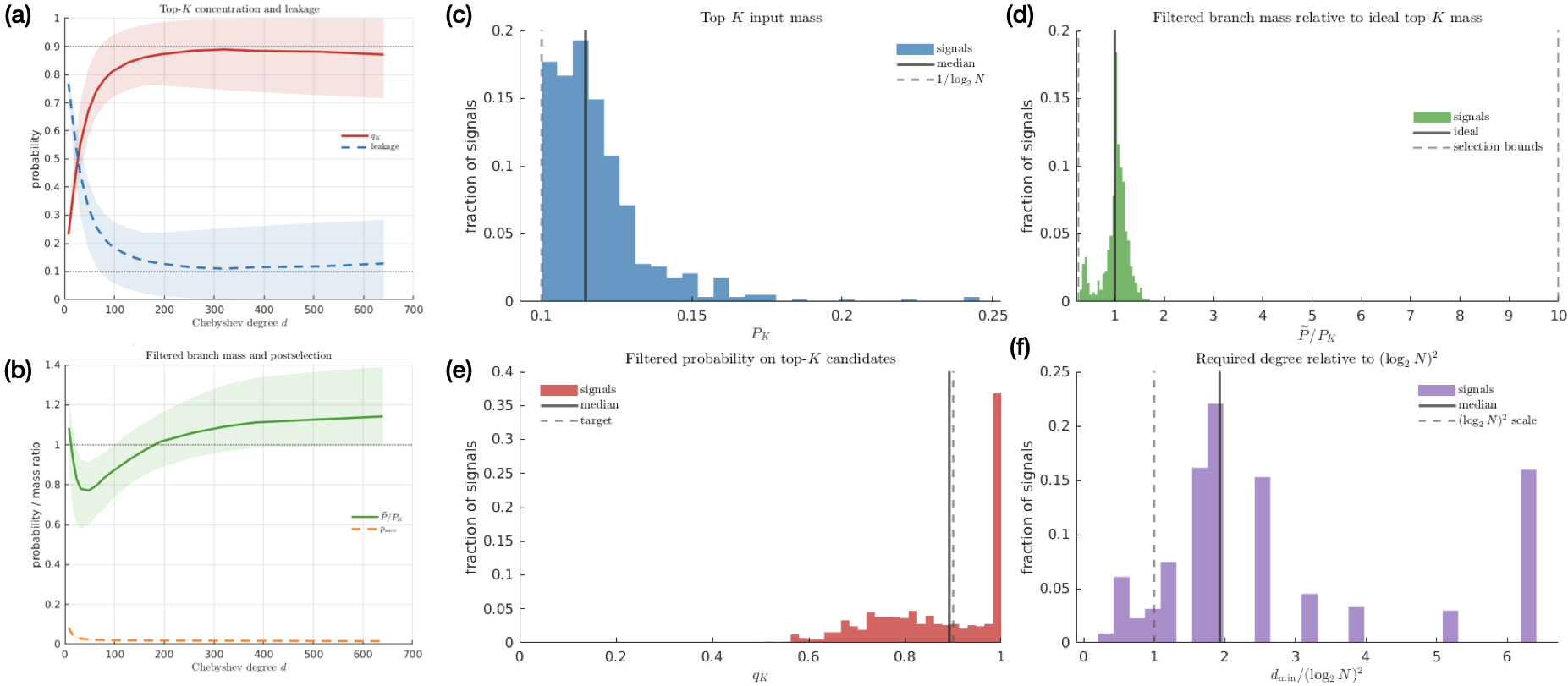}
\caption{\justifying \textbf{Ensemble validation of QG1 over 576 turbulent velocity-gradient signals at $N=1024$.}
Solid curves denote ensemble medians and shaded bands denote interquartile ranges.
(a) Probability $q_K$ retained on the leading candidate set increases with polynomial degree while the leakage $1-q_K$ decreases, approaching a stable high-concentration regime.
(b) The filtered branch mass $P_e$ remains comparable to the corresponding input candidate mass $P_K$ as the filter converges; the associated raw postselection probability is $p_{\rm succ}=P_e/C_d^2$.
(c) Distribution of the input candidate mass $P_K$. Most signals satisfy the target $1/\log_2N=0.1$.
(d) Distribution of $P_e/P_K$, concentrated near unity, showing that the finite-degree transformation preserves the intended extreme-event mass.
(e) Distribution of the filtered candidate concentration, with median $q_K=0.89282$.
(f) Minimum accepted degree normalized by $(\log_2N)^2$, with median $d_{\min}/(\log_2N)^2=1.92$.
Together these results test the retained-mass, candidate-compression, and finite-degree conditions that underlie the structured-search guarantee.}
\label{fig:qg1_degree_main}
\end{figure*}

\section{Discussion of results}
\label{sec:discussion}

This work shows how a global extremum can be extracted efficiently from amplitude-only quantum data as long as the extreme-event tail carries sufficient normalized probability. The underlying Chebyshev-walk primitive (see below) applies the nonlinear transformation $u_i\mapsto p_d(u_i)u_i$ without the need first for estimating the individual amplitudes. Its central algorithmic role is to isolate the coarse candidates from the exact ordering: the global stage resolves only the designed transition width and compresses the problem to $O(\log N)$ candidates, while the fine value and sign information are obtained only later. Global completeness then guarantees that this compression has not discarded the best candidate. QG2 supplies optional contrast and localization refinements, while the QG1 route and explicit enumeration remain rigorous fallbacks. The same framework covers the four extrema objectives described in Appendix~\ref{sec:Candidates,enum,cert,sign,cases}.

The Chebyshev-walk transformation is given by
\begin{equation}
p_d(D_\Psi)|\Psi\rangle
=
\sum_i p_d(u_i)u_i|i\rangle ,
\label{eq:discussion_transform}
\end{equation}
which applies a bounded nonlinear function directly to unknown quantum amplitudes using only the state-preparation circuit and its coherent inverse. This provides a general primitive for nonlinear processing of amplitude-encoded quantum data, although extrema filtering is the application considered here.

The second step in QG1 
suppresses the broad background while deliberately leaving a finite transition region unresolved. Consequently, the transition width controls the polynomial degree, while the retained mass controls candidate-state preparation; the final top-two spacing enters only after the signal has been reduced to $O(\log N)$ addresses. This is stronger than preparing or sampling a thresholded tail. The surviving addresses must still be recovered, their values and signs determined when required, and the possibility of a larger unobserved event must be excluded. These steps are included in the present algorithm. QG2, on the other hand, provides optional local-background subtraction and hierarchical localization, while the original QG1 route remains a rigorous fallback. This construction covers the largest magnitude, largest positive value, and most negative value. The smallest-magnitude problem requires additional preparation of the magnitude state after subtracting the background, as described in Appendix~\ref{app:four_extrema}, Corollary~\ref{cor:min_abs}.

\subsection{Structure as a quantum computational resource}
\label{sec:structure_advantage}
{For a physical signal $g_i$ encoded as
$u_i=g_i/(\sqrt{N}g_{\rm rms})$, define
$S_{\rm in}=g_{\rm pk}/g_{\rm rms}$, where
$g_{\rm pk}=\max_i|g_i|$. Let
$K_R:=|R_\Gamma|$ denote the number of safely retained
amplitudes and order the signal magnitudes as
$|g_{(1)}|\geq|g_{(2)}|\geq\cdots$. Define
$\rho_R^2:=K_R^{-1}\sum_{j=1}^{K_R}
(|g_{(j)}|/g_{\rm pk})^2$. Since $R_\Gamma$ is a
magnitude-threshold set, its probability mass is
\begin{equation}
P_R
=
\frac{K_R\rho_R^2S_{\rm in}^2}{N}.
\end{equation}
For the representative condition $a=1$, used in the numerical
study, $P_R\geq1/\log_2N$ is therefore guaranteed by
\begin{equation}
S_{\rm in}
\geq
S_\star
:=
\frac{\sqrt{N}}
{\rho_R\sqrt{K_R\log_2N}}.
\label{eq:structure_condition}
\end{equation}
The derivation is given in Appendix~\ref{app:structure_crossover}.}

\eqref{eq:structure_condition} identifies the physical resource behind the acceleration of the steps involved. The extreme events must be few in number, but together they must still carry enough normalized probability to be located efficiently; thus, rarity by itself is not enough.  In the amplitude-only model, this concentration supplies part of the information that a marking oracle would otherwise provide externally. Physical structure therefore enters the computation directly: it determines whether the important part of the signal can be isolated without examining the entire state.

The turbulence calculations show that such structure occurs naturally without having to consider planted examples. For the 576 intermittent gradient signals extracted from a fully resolved $1024^3$ Navier--Stokes simulation, the median leading-event mass exceeds the inverse-logarithmic target, QG1 reduces a $1024$-point signal to five candidates at the median, and successful filtering increases the probability of observing the true maximum by nearly a factor of eight. These observations test the retained-mass, candidate-count, and concentration mechanisms underlying Theorem~\ref{thm:main}, while the asymptotic scaling follows from the analytical bounds. 

It is worth noting that the task considered here is thus not merely generic anomaly detection, but the extraction of events central to one of the fundamental questions of turbulence and fluid dynamics.

The competition between the increasing problem size and increasing intermittency can be expressed by
\begin{equation}
N_{\rm db}({\rm Re})=C_N{\rm Re}^{\alpha},
\qquad
S_{\rm in}({\rm Re})=C_S{\rm Re}^{\beta}.
\label{eq:Re_scaling}
\end{equation}
Substitution into \eqref{eq:structure_condition} gives
\begin{equation}
C_S\mathrm{Re}^{\beta}
\geq
\frac{
\sqrt{C_N}\,\mathrm{Re}^{\alpha/2}
}{
\rho_R\sqrt{
K_R\log_2(C_N\mathrm{Re}^{\alpha})
}
}.\label{eq:Re_structure_condition}
\end{equation}

Increasing ${\rm Re}$ therefore enlarges the resolved search space while also strengthening the intermittent events that can make candidate isolation efficient. The resulting certified-search scalings and nominal crossover estimates are summarized in Table~\ref{tab:complexity_main}, with their derivation and limitations given in Appendix~\ref{app:structure_crossover}.
\begin{table*}
\centering
\caption{\textbf{Single-candidate-state preparation scaling and nominal preparation-dominated crossovers for representative flow-resolved databases.}
Here $L=\log_2N_{\rm db}$ and $\Delta_{\rm val}$ denotes the normalized probability or calibrated-value gap used for final value estimation. The scan and GDH columns use the same value-precision model; with exact digital values, their factors $1/\Delta_{\rm val}$ are removed. The present-method column shows the cost of one amplitude-amplified QG1 candidate-state preparation together with the final $O(L/\Delta_{\rm val})$ precision term. It does not include the repeated candidate-state factor $M_{\rm glob}$ required for complete candidate recovery. The two crossover columns therefore compare the preparation-dominated scaling only and serve as asymptotic guideposts rather than total certified-search crossover claims. Constants $C_d$, calibration and confidence factors, and the accuracy dependence are suppressed.}
\label{tab:complexity_main}
\resizebox{0.995\textwidth}{!}{%
\begin{tabular}{lllllll}
Case & $N_{\rm db}({\rm Re})$ & Raw scan & GDH + comparator &
Present method + certification & Beats raw scan & Beats ideal Grover scale \\
\midrule
1D Burgers &
${\rm Re}$ &
$O({\rm Re}/\Delta_{\rm val})$ &
$O({\rm Re}^{1/2}/\Delta_{\rm val})$ &
$\widetilde O(C_dL^2+L/\Delta_{\rm val})$ &
${\rm Re}\gtrsim16$ &
${\rm Re}\gtrsim6.55\times10^4$ \\

2D Navier--Stokes &
${\rm Re}$ &
$O({\rm Re}/\Delta_{\rm val})$ &
$O({\rm Re}^{1/2}/\Delta_{\rm val})$ &
$\widetilde O(C_dL^2+L/\Delta_{\rm val})$ &
${\rm Re}\gtrsim16$ &
${\rm Re}\gtrsim6.55\times10^4$ \\

3D HIT line &
${\rm Re}^{3/4}$ &
$O({\rm Re}^{3/4}/\Delta_{\rm val})$ &
$O({\rm Re}^{3/8}/\Delta_{\rm val})$ &
$\widetilde O[C_d(\tfrac34\log_2{\rm Re})^2+
(\tfrac34\log_2{\rm Re})/\Delta_{\rm val}]$ &
${\rm Re}\gtrsim4.03\times10^1$ &
${\rm Re}\gtrsim2.64\times10^6$ \\

3D HIT slice &
${\rm Re}^{3/2}$ &
$O({\rm Re}^{3/2}/\Delta_{\rm val})$ &
$O({\rm Re}^{3/4}/\Delta_{\rm val})$ &
$\widetilde O[C_d(\tfrac32\log_2{\rm Re})^2+
(\tfrac32\log_2{\rm Re})/\Delta_{\rm val}]$ &
${\rm Re}\gtrsim6.35$ &
${\rm Re}\gtrsim1.63\times10^3$ \\

3D HIT volume &
${\rm Re}^{9/4}$ &
$O({\rm Re}^{9/4}/\Delta_{\rm val})$ &
$O({\rm Re}^{9/8}/\Delta_{\rm val})$ &
$\widetilde O[C_d(\tfrac94\log_2{\rm Re})^2+
(\tfrac94\log_2{\rm Re})/\Delta_{\rm val}]$ &
${\rm Re}\gtrsim3.43$ &
${\rm Re}\gtrsim1.38\times10^2$ \\
\bottomrule
\end{tabular}%
}
\end{table*}

{The present homotopy calculations reach nonlinear flow regimes of
order $O(10^2)$, which is also the order of the nominal
preparation-dominated extrema-search crossover for an equivalent three-dimensional resolved volume. The higher-Reynolds-number JHTDB is explored separately as a demanding stress test of the extrema-search stage. The detailed compatibility analysis is given in Appendix~\ref{sec:end_to_end}}.

\subsection{End-to-end and verifiable quantum advantage}
\label{sec:verifiable_advantage}

An asymptotically efficient subroutine alone does not establish a useful scientific quantum advantage. The physical state must be generated efficiently, the desired observable must be extracted without reconstructing the entire solution, and the returned result should be verifiable efficiently by a classical witness. A central advance here is that these requirements can be connected within one workflow for a nonlinear physical problem.

For fluid dynamics, the state-preparation unitary can be supplied by the quantum homotopy algorithm of ref.~\cite{bharadwaj2025homotopy}. This prepares a normalized nonlinear PDE solution with near-optimal dependence on the resolved problem size under its stated convergence, sparsity, conditioning, accuracy, and postselection assumptions. Derived quantities such as velocity gradients can subsequently be generated coherently using the compact LCU operations of ref.~\cite{bharadwaj2024compact}. The complete pipeline is therefore
\begin{equation}
\text{nonlinear PDE}
\longrightarrow
|\Psi(T)\rangle
\longrightarrow
|\Psi_\nabla(T)\rangle
\longrightarrow
i^\star .
\label{eq:e2e_pipeline}
\end{equation}
The quantum state remains coherent until the required extremal information has been extracted, avoiding an intermediate reconstruction of all $N$ field values. Under the combined assumptions of these algorithms, the full workflow retains logarithmic or polylogarithmic dependence on the resolved problem size, with the physical sparsity, convergence, conditioning, postselection, and final-precision costs remaining explicit. The complete accounting is given in Appendix~\ref{sec:end_to_end}.

Various specifics of the output can also be checked efficiently. Global completeness guarantees that the true extremum is recovered with controlled failure probability rather than merely being the largest sampled candidate. Level I of Proposition~\ref{prop:verification} then provides an $O(\log N)$ classical check using the recovered candidate probabilities and missing filtered mass. At Level II, the complete computation is reduced to a bounded-error decision problem, and can be placed within the classical verification protocol of ref.~\cite{mahadev2018classical}. When the circuit size and security parameter are polylogarithmic in $N$, the verifier cost is also similar. This use of verification differs from the work that establishes verifiable quantum advantage without relying on physical structure~\cite{yamakawa2024verifiable}. Here the computational advantage comes from a concrete scientific problem and its measurable structure; the verification layer ensures that the reported quantum result can also be checked independently.

During the final compilation of this manuscript, we became aware of the recent work of Guo et al.~\cite{Guo2026RareEvents}. 
At a high level, the two studies address different questions. Ref.~\cite{Guo2026RareEvents} asks how a quantum state can be transformed so that rare events below a prescribed probability threshold can be sampled more efficiently. The present work asks the broader question of whether a physically important observable can be extracted from a large quantum-encoded scientific problem with a rigorous, end-to-end, and verifiable quantum advantage. This requires considerably more than isolating a rare-event tail. The present work addresses some of the following
questions within one framework. Can the relevant physical state itself be generated efficiently rather than assumed as freely loaded data? Can a specific global extremum be found rather than merely sampling an event from a selected subset? Can the search be performed without value, comparison, sign, marking, or threshold oracles? Can one prove that no larger event was missed during compression? Can the final physical result be checked efficiently by a classical verifier? Can one identify measurable properties of the underlying physical system that determine when the quantum advantage exists, and demonstrate those conditions on a realistic nonlinear problem? Further, it connects nonlinear PDE state generation, structure-dependent candidate isolation, certified global extrema recovery, physical validation on turbulent Navier--Stokes data, and classical verification of the resulting computation. 
Therefore, the computational objective, guarantees, physical scope, and end-to-end quantum-advantage question addressed here are fundamentally different.

\section{Conclusions}
\label{sec:conclusions}

Taken together, these results establish a complete pathway from nonlinear quantum-state generation to a globally certified and classically verifiable physical observable. The principal advances are a bounded nonlinear transformation of unknown amplitudes, certified extrema search using amplitude-only access, a measurable structure criterion that determines when the search is accelerated, and validation on strongly intermittent Navier--Stokes data. State preparation, extrema search, physical observable extraction, and verification are thus treated as components of a single scientific computation rather than as independent quantum subroutines. (For the present example, the classical turbulence data had to be converted to a quantum state, but this overhead is not intrinsic to the present algorithm.)

The framework is not limited to turbulence. It applies whenever the events of interest occupy a sufficiently concentrated portion of an efficiently prepared quantum state, including signals arising from wave phenomena, imaging, geophysics, and other coherently generated scientific data. A particularly natural extension is direct integration with quantum sensing. If a quantum sensor can coherently transduce a physical field into quantum amplitudes, the resulting state could be supplied directly to the extrema-search circuit, allowing rare and rapidly evolving events to be identified before the full field is measured, stored, and reconstructed classically~\cite{allen2025quantum,khan2025quantum}. Such a sensor-to-computer architecture would couple quantum measurement to quantum processing at precisely the stage where the physically significant information is generated.

More broadly, our results suggest that useful quantum advantage in scientific computing need not be determined by problem size alone. It can emerge when the underlying physical dynamics concentrate the information of interest into a small but sufficiently probable portion of a quantum state, and when an algorithm can exploit that organization without destroying it through full readout. In this sense, the deeper computational resource is the conversion of physical structure into algorithmic advantage.

\acknow{
S.S.B thanks Scott Aaronson for insightful discussions and comments on this work. The authors would also wish to thank H.A. Vinutha, J\"org Schumacher and Akash Rodhiya for helpful discussions. The work was supported by the NSF award 24-600, Number 2534977, under the title, TRAILBLAZER Quantum Computing and Machine Learning for Fluid Dynamics Research. The data used in this paper are obtained from the Johns Hopkins University Turbulence Data Base. The authors also thank New York University for its support of this work.}

\dataavail{All the data is included in the manuscript. The code will be made available following any reasonable request made to the authors.}

\showacknow{} 

\appendix

\setcounter{lemma}{0}
\setcounter{corollary}{0}
\setcounter{remark}{0}
\setcounter{proposition}{0}
\setcounter{theorem}{0}
\renewcommand*{\theHlemma}{app.\arabic{lemma}}
\renewcommand*{\theHcorollary}{app.\arabic{corollary}}
\renewcommand*{\theHremark}{app.\arabic{remark}}
\renewcommand*{\theHproposition}{app.\arabic{proposition}}
\renewcommand*{\theHtheorem}{app.\arabic{theorem}}

\renewcommand{\ssb}[1]{{\color{black}#1}}
\renewcommand{\ket}[1]{\lvert #1\rangle}
\renewcommand{\bra}[1]{\langle #1\rvert}
\renewcommand{\braket}[2]{\langle #1\vert #2\rangle}
\providecommand{\norm}[1]{\left\lVert #1\right\rVert}
\providecommand{\abs}[1]{\left\lvert #1\right\rvert}
\renewcommand{\polylog}{\operatorname{polylog}}
\renewcommand{\sgn}{\operatorname{sgn}}
\providecommand{\eO}{\widetilde{O}}
\providecommand{\R}{\mathbb{R}}

\section{Grover--D\"urr--H\o yer extrema search}
\label{sec:gdh}

We revisit the well-known $O\left(\sqrt{N}\right)$ benchmark \ssb{result} for quantum extrema \ssb{searching} \cite{grover1996fast,durr1996quantum}. \ssb{It may be useful to note that,} this result assumes that the algorithm can coherently mark indices whose values exceed a \ssb{set} threshold  \ssb{whose value is sequentially updated}.  That assumption is stronger than the amplitude-only access model used in the present work, \ssb{and this} must be kept explicit when the two methods are compared.

Let $\{v_i\}_{i=0}^{N-1}$ denote the scalar values being searched. A classical certified search requires $N-1=O(N)$ comparisons, since every entry must be checked against the current best value. Grover search instead assumes coherent access to a marked set $M\subseteq[N]$ of size $m=|M|$.  Let $\ket{G}$ and $\ket{B}$ denote the normalized, \ssb{uniform states encoding the locations} of the marked and unmarked indices, respectively.  The \ssb{normalized} initial state can then be written as
\begin{equation}
\begin{aligned}
\ket{s}
&=
\sqrt{\frac{m}{N}}\ket{G}
+
\sqrt{\frac{N-m}{N}}\ket{B}\\
&=
\sin\theta\ket{G}+\cos\theta\ket{B},
\qquad
\sin^2\theta=\frac{m}{N}.
\end{aligned}
\label{eq:A1}
\end{equation}
\ssb{The phase oracle $O_M$ marks the set $M$ by a sign flip operation, $O_M|i\rangle=(-1)^{[i\in M]}|i\rangle$, equivalently reflecting the marked component. This is followed by} $D_s=2\ket{s}\!\bra{s}-I$\ssb{, a second} reflection about the initial state \ssb{$\ket{s}$.}  Their product is a rotation in the $\{\ket{G},\ket{B}\}$ plane \ssb{, rotating the state closer to $\ket{G}$}. Therefore $r$ such operations would yield,
\begin{equation}
(D_sO_M)^r\ket{s}
=
\sin[(2r+1)\theta]\ket{G}
+
\cos[(2r+1)\theta]\ket{B}.
\label{eq:A2}
\end{equation}
Choosing $(2r+1)\theta\simeq\pi/2$ therefore gives
\begin{equation}
r=O\!\left(\sqrt{\frac{N}{m}}\right),
\label{eq:A3}
\end{equation}
and hence the familiar $O(\sqrt{N})$ cost \cite{grover1996fast} for a single marked item.

D\"urr--H\o yer maximum finding \cite{durr1996quantum} converts this marked-item search into an extremum search by using the current candidate $j$ as a threshold.  It defines
\begin{equation}
M_j=\{i:v_i>v_j\},
\qquad
O_j\ket{i}=(-1)^{[v_i>v_j]}\ket{i}.
\label{eq:A4}
\end{equation}
If the current candidate has rank $r$, then $|M_j|=r-1$ \ssb{(potential candidates for extrema),} and Grover search finds a better candidate using $O(\sqrt{N/(r-1)})$ queries.  A successful update reduces the rank by a constant factor with constant probability.  \ssb{Summing the costs over such iterative dyadic searches gives}
\begin{equation}
\sum_{k=0}^{\lceil\log_2N\rceil}O\!\left(2^{k/2}\right)
=O(\sqrt{N}),
\label{eq:A5}
\end{equation}
which is the expected Grover--D\"urr--H\o yer query complexity \cite{grover1996fast,durr1996quantum}.  Intuitively, the successive rank scales cost approximately as $1,\sqrt{2},2,\ldots,\sqrt{N}$. Now because this is a geometric sequence, \ssb{the sum scales as $O(\sqrt{N})$}. The important qualification is that \eqref{eq:A5} assumes that the value-comparison oracle $O_j$ is already available coherently.  In the present work, the input is supplied only through
\begin{equation}
U_S\ket{0^n}
=
\sum_{i=0}^{N-1}u_i\ket{i},
\label{eq:A6}
\end{equation}
which does not provide either $\ket{i}\ket{0}\mapsto\ket{i}\ket{u_i}$ or the comparison phase $(-1)^{[u_i>u_j]}$.  \ssb{State-preparation access exposes the complete vector coherently, but
not as a numerical lookup table; we do not assume QRAM access to its
entries.}

If the comparator is synthesized from amplitude-only access using
singleton-probability estimation with precision
$\epsilon_{\rm cmp}$, the corresponding benchmark overhead becomes
$O(\sqrt{N}/\epsilon_{\rm cmp})$.  \ssb{If $p_{(1)}$ and $p_{(2)}$ are the top two, largest absolute amplitudes and writing,}
\[
\delta_p:=p_{(1)}-p_{(2)},
\qquad
p_i:=|u_i|^2 ,
\]
\ssb{and if}
$\epsilon_{\rm cmp}=O(\delta_p)$\ssb{, then certifying the global extrema, requires at least $O(1/\delta_p)$} queries and therefore gives
 $O(\sqrt{N}/\delta_p)$, up to confidence factors.

Thus Grover--D\"urr--H\o yer and the present amplitude-only algorithm
solve the same high-level search task under different primitive access
assumptions: the former begins with a coherent comparison oracle,
whereas the latter must first create the target amplitude state and then produce a candidate-concentrated representation, on which extrema certification
becomes efficient.

\section{Access model}
\label{sec:access}

The input access model is intentionally kept weak \ssb{(not requiring the typical black-box oracles)}, to enable practical and near-term utility of the proposed algorithm. The proposed algorithm includes within itself, a well-defined quantum subroutine that can prepare and coherently reuse a quantum state whose amplitudes carry the signal. However we do not assume random access to the individual amplitudes as classical numbers. This distinction is essential, because a state preparation exposes the whole vector coherently, whereas a value oracle would make each entry directly comparable and would change the search model. In essence, the proposed algorithm comes with a quantum state preparation algorithm (see Secs.~\ref{subsec:State_Preparation_Homotopy} and~\ref{subsec:State_Preparation_Sensor}) that prepares the target quantum signal with near-optimal complexity guarantees. 

\noindent We use a real, normalized signal state
\begin{equation}
\ket{\Psi}
=
\sum_{i=0}^{N-1}u_i\ket{i},
\qquad
\sum_{i=0}^{N-1}u_i^2=1,
\qquad
N=2^n.
\label{eq:B1}
\end{equation}

\ssb{The algorithm is formulated for an ideal preparation unitary
$U_S$, its inverse, and their controlled versions, satisfying
\begin{equation}
U_S|0^n\rangle=|\Psi\rangle .
\end{equation}
For a finite-accuracy implementation, let $\widetilde U_S$
satisfy
\begin{equation}
\|\widetilde U_S-U_S\|\leq\epsilon_S
\end{equation}
in operator norm, with the same accuracy requirement imposed on
the corresponding controlled implementation. Since
$\|\widetilde U_S^\dagger-U_S^\dagger\|
=\|\widetilde U_S-U_S\|$, the inverse obeys the same bound.
For a computation containing $Q$ coherent uses of these
preparation primitives, the standard hybrid argument gives an
accumulated state error of at most $O(Q\epsilon_S)$; the
implementation accuracy is chosen so that this contribution lies
within the prescribed total error budget. The gate cost of one
state-preparation invocation is denoted by $C_S$.}

\subsection{State preparation}
\label{subsec:State Preparation_General}
We consider two physically distinct sources of $U_S$. (1) A quantum PDE solver that prepares a flow solution, and (2) A coherent quantum sensor that transduces a physical field directly into amplitudes. The former route is the primary choice for this present work, because it gives the entire algorithmic pipeline an end-to-end quantum advantage with near-optimal complexity guarantees. The latter option is proposed here, only from a forward looking perspective for integrating quantum sensors with quantum computing devices. The point of the examples is to identify concrete state preparation oracles, with complexity and optimality guarantees in which the target signal can be efficiently prepared natively as a quantum state, with only $O(\polylog{N})$ queries, for an $O(\log N)$ size classical inputs that prepares a state of size $O(N)$.

\subsubsection{Homotopy linearization and nonlinear PDE state preparation}
\label{subsec:State_Preparation_Homotopy}

To prepare quantum states proportional to the solution of a nonlinear PDE, we use the quantum homotopy solver of ref.~\cite{bharadwaj2025homotopy}. It is the state-of-the-art nonlinear PDE solver, in terms of the combined complexity, admissible nonlinearity, and near-term implementation. The method is end-to-end and near-optimal in the restricted sense that its dominant dependence on the spatial problem size, number of time steps, and accuracy is logarithmic or polylogarithmic, while retaining explicit sparsity, conditioning, and postselection factors~\cite{bharadwaj2025homotopy,bharadwaj2024compact}.

For example, after spatial discretization, a typical quadratic nonlinear PDE can be written schematically as
\begin{equation}
\frac{d u}{dt}
=
D_1 u^{\otimes 2}
+
D_2 u
+
D_0(t),
\label{eq:B5}
\end{equation}
where $D_1$ represents the quadratic nonlinear coupling, $D_2$ the linear dynamics, and $D_0$ source or boundary contributions. For example, in Burgers flow $u_t+u u_x=\nu u_{xx}$, the nonlinear term is generated by the first-derivative operator and the linear dissipative term by the second derivative. It is also useful to note that, as shown in ref.~\cite{bharadwaj2025homotopy}, the algorithm's applicability extends to a wider range of the nonlinearity forms, beyond the example PDE noted above.

The homotopy method does not attack \eqref{eq:B5} directly. Instead, it starts from a known base solution $u_0$ of a chosen linear problem and continuously deforms it into the nonlinear solution. Writing the nonlinear PDE as $\mathcal N[u]=0$, the deformation is defined by~\cite{bharadwaj2025homotopy}
\begin{equation}
(1-q)\mathcal L[\phi(q)-u_0]
=
q\bar h\,\mathcal N[\phi(q)],
\qquad
q\in[0,1],
\label{eq:B6}
\end{equation}
so that $\phi(0)=u_0$ and $\phi(1)=u$. Expanding
\begin{equation}
\phi(q)
=
u_0+\sum_{j=1}^{M}q^j u_j
\label{eq:B7}
\end{equation}
turns the original nonlinear problem into a sequence of linear deformation equations. The terms $u_1,u_2,\ldots,u_M$ are successive nonlinear corrections to the chosen base solution. Their sum at $q=1$ gives the truncated nonlinear solution $u^{(M)}=u_0+\sum_{j=1}^{M}u_j$.

The important quantum step is that these corrections are not solved one after another. The deformation variables and only the additional closure variables generated by their nonlinear products are assembled into a single sparse lifted system,
\begin{equation}
\frac{d v}{dt}=A^{(M)}v ,
\label{eq:B8}
\end{equation}
which is solved coherently using the compact time-marching quantum algorithm of ref.~\cite{bharadwaj2024compact}. The lift is therefore adapted to the chosen base physics and homotopy recursion rather than containing the full hierarchy of all tensor powers. Once the homotopy order $M$ is fixed, the retained lifted system is closed by construction~\cite{bharadwaj2025homotopy}. The output is directly the normalized nonlinear solution state
\begin{equation}
|u^{(M)}(T)\rangle
=
\frac{u^{(M)}(T)}
{\|u^{(M)}(T)\|_2},
\label{eq:B9}
\end{equation}
which is precisely the amplitude-state input required by the extrema-search algorithm.

For the compact lifted solver, the resulting complexity is~\cite{bharadwaj2025homotopy,bharadwaj2024compact}
\begin{equation}
\widetilde O\!\left[
s_H\epsilon_{\rm LCU}
\log(N_gM^2\tau)\,
\operatorname{polylog}(\epsilon_Q^{-1},\kappa_H)
\right],
\label{eq:B10}
\end{equation}
where $N_g$ is the spatial grid size, $\tau$ is the number of time steps, $s_H$ is the lifted sparsity, $\kappa_H$ is the conditioning parameter, and $\epsilon_Q$ is the quantum-solver error. The underlying compact LCU construction is designed for shallow, near-term matrix operations and avoids treating nonlinear PDE state preparation as a free loading oracle~\cite{bharadwaj2024compact}.

The solver also admits a physically interpretable nonlinear regime. For Burgers flow, the homotopy nonlinearity parameter satisfies $R_H=O(\mathrm{Re})$ and its certified finite-time window scales as $T_H=O(\mathrm{Re}^{-1})$~\cite{bharadwaj2025homotopy}. For the physically rescaled comparison developed in Ref.~\cite{bharadwaj2025homotopy}, the corresponding Carleman convergence window contracts more rapidly, as $T_C=O({\rm Re}^{-5/2})$, whereas the homotopy window scales as $T_H=O({\rm Re}^{-1})$. Numerical demonstrations reach $R_H=O(10^2)$ while retaining convergent homotopy solutions. This is compatible with the regime relevant to the present extrema-search algorithm, whose structure-enabled advantage also appears at physically resolved problem sizes corresponding to nonlinear flow regimes of the same order. A more detailed comparison of these operating windows is given in Sec.~\ref{subsec:Re_overlap}.

\subsubsection{Coherent sensor preparation}
\label{subsec:State_Preparation_Sensor}
A complementary route is to generate the state directly through a coherent interaction with the physical field.  In a typical quantum sensor, the field imprints a phase on a probe during an interrogation time $t_s$, which is subsequently converted into a population.  If $i$ labels a coherently addressable location, mode, frequency, or other candidate, this transduction may be represented as
\begin{equation}
V_g\ket{i}\ket{0}_q
=
\ket{i}\left[
\cos(\phi_i)\ket{0}_q+\sin(\phi_i)\ket{1}_q
\right],
\qquad
\phi_i=\gamma t_sg_i,
\label{eq:B11}
\end{equation}
where $g_i$ is the corresponding field observable, $\gamma$ is the sensor coupling strength, and $q$ is a sensor ancilla.  Applying $V_g$ to a uniform superposition and coherently selecting the $\ket{1}_q$ branch prepares
\begin{equation}
\ket{\Psi_g}
=
\frac{1}{\mathcal N_g}
\sum_{i=0}^{N-1}\sin(\gamma t_sg_i)\ket{i},
\qquad
\mathcal N_g^2
=
\sum_i\sin^2(\gamma t_sg_i).
\label{eq:B12}
\end{equation}
In the weak-response regime, $\sin(\gamma t_sg_i)\simeq\gamma t_sg_i$, and the amplitudes are therefore proportional to the sensed field.  More generally, the sensor response need only preserve or enhance its extreme tail.  The field can then enter the extrema-search circuit without first being measured, stored, and reloaded as $N$ classical values \cite{allen2025quantum,khan2025quantum}.

For this state to provide the $U_S$ access required by the algorithm, the sensor ancilla and all work registers must be coherently uncomputed.  The resulting preparation, its inverse, and their controlled versions must be repeatable during each circuit attempt.  This requires coherent global or multimode access to the candidate labels.  If the sensor must instead interrogate $N$ independent locations sequentially, the preparation cost is $\Theta(N)$ and no end-to-end asymptotic advantage follows.

For a time-dependent flow, the field need not remain fixed throughout an arbitrarily long experiment; the relevant requirement is that a sensing--processing cycle be short compared with the event evolution.  With Kolmogorov time $\tau_\eta=(\nu/\langle\varepsilon\rangle)^{1/2}$, representative turbulent and solid-state sensing times can occupy overlapping microsecond-to-millisecond ranges \cite{mordant2004experimental,do2024turbulence,arunkumar2023quantum}.  A useful operating regime is therefore
\begin{equation}
T_{\rm probe}\lesssim\tau_\eta,
\qquad
M\sim\frac{\tau_{\rm evt}}{T_{\rm probe}}>1,
\label{eq:B13}
\end{equation}
so that several short probes can be performed while an extreme event evolves over $\tau_{\rm evt}=O(\tau_\eta)$ or longer.

\subsection{Near-optimal iterative matrix--vector operations}
\label{subsec:LCU_Matrix_Ops}
For the PDE-solver route it is useful to record one concrete mechanism based on Linear Combination of Unitaries (LCU) for iterative, nonunitary matrix actions used across the manuscript \cite{bharadwaj2024compact}. These operations are not assumed to have unit cost, instead at most polylogarithmic in problem size $N$ and accuracy $\epsilon$.  For a nonunitary and non-Hermitian operator $M$, define
\begin{equation}
S=\frac{M+M^\dagger}{2},
\qquad
A=\frac{M-M^\dagger}{2}.
\label{eq:B14}
\end{equation}
The four-unitary construction used in the compact solver, for $\epsilon\to0$, is
\begin{equation}
\widetilde D_\epsilon
=
\frac{ie^{-i\epsilon S}-ie^{i\epsilon S}+e^{\epsilon A}-e^{-\epsilon A}}{2\epsilon}
=
M+O(\epsilon^2),
\label{eq:B15}
\end{equation}
with an $O(\epsilon^4)$ leading error after second-order Richardson extrapolation.  More generally, let
\begin{equation}
\widetilde M
=
\sum_{k=0}^{K-1}\alpha_kU_k,
\qquad
\alpha=\sum_k|\alpha_k|,
\qquad
K\in\{2,4\}
\label{eq:B16}
\end{equation}
then the LCU circuit yields
\begin{equation}
U_D\ket{0^a}\ket{\varphi}
=
\ket{0^a}\frac{\widetilde M\ket{\varphi}}{\alpha}
+\ket{\Phi^\perp},
\qquad
p_D
=
\frac{\norm{\widetilde M\ket{\varphi}}_2^2}{\alpha^2}.
\label{eq:B19}
\end{equation}
The same construction implements $M\ket{x}$, repeated multiplication $M^\tau\ket{x}$, and, using a convergent truncated Neumann expansion, $M^{-1}\ket{x}$ and $(M^{-1})^\tau\ket{x}$.  For the one-shot explicit and implicit algorithms, the total gate complexity is
\begin{equation}
\begin{aligned}
C_{\rm LCU}
=
O\!\Biggl[&
 s\epsilon_{\rm LCU}\log_2(N_g\tau)\\
&\times
 \polylog\!\left(\frac{\epsilon_{\rm LCU}}{\epsilon_U},\frac{1}{\epsilon_N}\right)\\
&\times
 \left(\log_2\frac{1}{\kappa-1}\right)^{-3}
\Biggr].
\end{aligned}
\label{eq:B20}
\end{equation}
Here $s$ is the sparsity of $M$, $\epsilon_{\rm LCU}$ controls the LCU decomposition, $\epsilon_U$ is the unitary-synthesis error, $\epsilon_N$ is the Neumann-truncation error, and $\kappa$ is the condition number of the scaled matrix.  The inverse construction assumes convergence of the Neumann expansion.  For fixed-order sparse discretizations, $s=O(1)$ and $\epsilon_{\rm LCU}$ can remain close to unity, so the dependence on the matrix--iteration size $N_g\tau$ is near optimal up to the displayed logarithmic and conditioning factors \cite{bharadwaj2024compact}.

\subsection{Amplitude estimation and measurement}
\label{subsec:QAE}
For a preparation circuit
\begin{equation}
U_M\ket{0}
=
\sqrt{p_\Omega}\ket{\mathrm{good}}
+
\sqrt{1-p_\Omega}\ket{\mathrm{bad}},
\label{eq:B21}
\end{equation}
where
\begin{equation}
p_\Omega
=
\bigl\|P_\Omega U_M|0\rangle\bigr\|_2^2
=
\sum_{i\in\Omega}|u_i|^2 .
\label{eq:B22}
\end{equation}
quantum amplitude estimation (QAE) returns $|\widetilde p_\Omega-p_\Omega|\le\epsilon_{\rm QAE}$ with failure probability at most $\delta_{\rm QAE}$ using
\begin{equation}
O\!\left(\epsilon_{\rm QAE}^{-1}\log\delta_{\rm QAE}^{-1}\right)
\label{eq:B23}
\end{equation}
controlled calls to $U_M$ and $U_M^\dagger$ \cite{brassard2002quantum,grinko2021iterative}.  For $\Omega=\{i\}$ this estimates $|u_i|^2$; it does not reveal the sign of a real amplitude.  The sign protocol is therefore postponed until after candidate localization in Appendix~\ref{sec:Candidates,enum,cert,sign,cases}.\\

The principal access model does \emph{not} assume the following -- QRAM or coherent lookup of classically stored data; a value oracle $O_u\ket{i}\ket{0}=\ket{i}\ket{u_i}$; a threshold/comparison oracle $O_t\ket{i}=(-1)^{[u_i>t]}\ket{i}$; a sign oracle $O_{\rm sgn}\ket{i}\ket{0}=\ket{i}\ket{\sgn(u_i)}$; or unit-cost access to arbitrary $M$, $M^\tau$, or $M^{-1}$.  Every time one of these stronger primitives is used for a comparison baseline or an optional transformed state, its cost is stated separately.

\section{State discrimination and threshold statistics}
\label{sec:No-Gos_Threshold_stat}

This appendix separates two separate concepts.  The first provides an information-theoretic no-go result, which shows that exact extrema search on completely arbitrary amplitude states cannot avoid dependence on a decisive amplitude gap.  The second is a benign preprocessing task: coarse mean and variance estimates can be used to place a scalar threshold without ever standardizing the amplitude operator itself.  Keeping these points separate explains why the fuzzy filter can avoid resolving the final top-two gap without contradicting lower bounds.

\subsection{State Discrimination Cost}
\label{subsec:State Discrimination Cost}
Let $a>b\ge0$ and let $\ket{\chi}$ be orthogonal to $\ket{0}$ and $\ket{1}$.  Choose $\norm{\chi}^2=1-a^2-b^2$ and define
\begin{equation}
\ket{\psi_0}=a\ket{0}+b\ket{1}+\ket{\chi},
\label{eq:C1}
\end{equation}
\begin{equation}
\ket{\psi_1}=b\ket{0}+a\ket{1}+\ket{\chi}.
\label{eq:C2}
\end{equation}
The maximum is at index $0$ for $\ket{\psi_0}$ and at index $1$ for $\ket{\psi_1}$.  With $\Delta=a-b$,
\begin{equation}
\braket{\psi_0}{\psi_1}
=2ab+\norm{\chi}^2
=1-(a-b)^2
=1-\Delta^2.
\label{eq:C3}
\end{equation}
Thus the correct maximizing index can change even while the two quantum states remain almost indistinguishable when $\Delta$ is small. For $T$ independent copies, the overlap is $(1-\Delta^2)^T$.  The optimal binary-discrimination error for equal priors is
\begin{equation}
P_{\rm err}^{(T)}
=
\frac12\left[
1-\sqrt{1-(1-\Delta^2)^{2T}}
\right].
\label{eq:C4}
\end{equation}
\eqref{eq:C4} is the Helstrom limit \cite{Helstrom1976}, which translates the remaining overlap of the two $T$-copy states into the error probability of the \emph{best possible} quantum measurement.  It is therefore not a statement about a particular QG1 circuit.  A constant bias over random guessing requires $(1-\Delta^2)^{2T}\le c<1$.  Since $\log(1-\Delta^2)=-\Theta(\Delta^2)$ for small $\Delta$,
\begin{equation}
T=\Omega(\Delta^{-2}).
\label{eq:C5}
\end{equation}

Coherent oracle access can improve this dependence quadratically, but cannot remove it.  Choose state-preparation unitaries $U_0,U_1$ satisfying $U_j\ket{0^n}=\ket{\psi_j}$ and differing only by the corresponding small rotation.  They may be chosen such that
\begin{equation}
\norm{U_0-U_1}
=O\!\left(\norm{\ket{\psi_0}-\ket{\psi_1}}\right)
=O(\Delta).
\label{eq:C6}
\end{equation}
For a $T$-query algorithm, replace its oracle calls one at a time from $U_0$ to $U_1$.  The triangle inequality gives a final-state distance at most
\begin{equation}
T\norm{U_0-U_1}=O(T\Delta).
\label{eq:C7}
\end{equation}
Constant-bias discrimination requires constant trace distance, hence
\begin{equation}
T=\Omega(\Delta^{-1}).
\label{eq:C8}
\end{equation}
This is the standard hybrid argument behind black-box quantum-search lower bounds \cite{Bennett1997}.  The conclusion is not that the present algorithm must resolve the final top-two gap during QG1.  Rather, no exact extrema algorithm can be gap independent on every possible amplitude state.  The present construction responds by asking the global filter a weaker question -- isolate a structured, inverse-logarithmic-mass candidate tail across a deliberately wider fuzzy band, and pay the fine value precision only after that tail has been compressed to $O(\log N)$ addresses.

\subsection{Mean and variance statistics for threshold selection}
\label{subsec:threshold_stat}

The main theorem is stated with an explicit lower fuzzy edge.  In data analysis one may instead select a small family of scalar thresholds from one-time estimates of the amplitude mean and variance.  Define
\begin{equation}
\mu_\Psi=\frac1N\sum_i u_i,
\qquad
\sigma_\Psi^2=\frac1N\sum_i(u_i-\mu_\Psi)^2.
\label{eq:C9}
\end{equation}
Normalization immediately gives
\begin{equation}
\sigma_\Psi^2
=
\frac1N\sum_i u_i^2-\mu_\Psi^2
=
\frac1N-\mu_\Psi^2.
\label{eq:C10}
\end{equation}
Moreover, with the uniform state $\ket{s}=H^{\otimes n}\ket{0^n}$,
\begin{equation}
\bra{0^n}H^{\otimes n}\ket{\Psi}
=
\frac1{\sqrt N}\sum_i u_i
=
\sqrt N\,\mu_\Psi.
\label{eq:C11}
\end{equation}
A sign-sensitive Hadamard test \cite{Ekert2002} therefore determines the sign and magnitude of the overlap, and amplitude estimation gives additive accuracy $\epsilon_{\rm stat}$ using $O(\epsilon_{\rm stat}^{-1}\log\delta_{\rm stat}^{-1})$ controlled state-preparation calls.  For the largest-magnitude problem, candidate-count control can be imposed independently of this statistical placement. Let $L=\log_2N$ and define
\begin{equation}
\lambda_{\kappa,N}
:=
\frac{1}{\sqrt{\kappa_{\rm cand}L}}.
\label{eq:C_lambda_kappa}
\end{equation}
The scanned lower threshold is restricted to $\lambda\geq\lambda_{\kappa,N}$. Then
\begin{equation}
1
=
\sum_i u_i^2
\geq
\sum_{i\in C_\Gamma}u_i^2
\geq
|C_\Gamma|\lambda^2,
\label{eq:C_candidate_norm_bound}
\end{equation}
which immediately gives
\begin{equation}
|C_\Gamma|
\leq
\frac{1}{\lambda^2}
\leq
\kappa_{\rm cand}L.
\label{eq:C15}
\end{equation}
Thus the logarithmic candidate-count condition can be enforced by the chosen lower edge. By contrast, the retained-mass condition $P_{R_\Gamma}\geq c_P/L$ is not a consequence of normalization or coarse mean and variance statistics. It is the physical structure condition that determines whether the polylogarithmic advantage guarantee applies. A practical threshold scan may then use
\begin{equation}
\tau_\ell
=
\widehat\mu_\Psi+\kappa_\ell\widehat\sigma_\Psi,
\qquad
\ell=0,\ldots,L_{\rm scan}-1,
\label{eq:C12}
\end{equation}
with $L_{\rm scan}=O(1)$. \ssb{The theorem requires only that
one selected threshold satisfy the stated candidate-count and
retained-mass conditions.  Any cost associated with identifying
such a threshold is denoted by $Q_{\rm cal}$ and is included
additively in the total complexity.}  For a prescribed scan of $L_{\rm scan}$ scalar thresholds, estimating the required calibration quantities to additive precision $\epsilon_{\rm cal}$ with total failure probability at most $\delta_{\rm cal}$ requires
\begin{equation}
Q_{\rm cal}
=
O\!\left[
\frac{L_{\rm scan}}{\epsilon_{\rm cal}}
\log\!\left(
\frac{L_{\rm scan}}{\delta_{\rm cal}}
\right)
\right]
\label{eq:C_Qcal}
\end{equation}
controlled state-preparation calls. The calibration need not resolve the final top-two amplitude spacing. It only has to place the threshold to accuracy smaller than the designed transition width, so we take $\epsilon_{\rm cal}=\Theta(\Gamma_N)$. With $L_{\rm scan}=O(1)$ and $\Gamma_N=(\log_2N)^{-3/2}$,
\begin{equation}
Q_{\rm cal}
=
\widetilde O\!\left[
(\log_2N)^{3/2}
\right].
\label{eq:C_Qcal_polylog}
\end{equation}
This cost is incurred once before the production candidate-state preparations and is therefore separate from $M_{\rm glob}Q_{\rm QG1}^{\rm AA}$ and from the final candidate-value certification cost.

The statistics in ~\eqref{eq:C12} choose only a \emph{scalar location} for the fuzzy transition.  QG1 continues to act directly on $D_\Psi$.  In particular, the algorithm does not block encode $(D_\Psi-\mu I)/\sigma$, so it does not incur a potentially large $1/\sigma$ normalization penalty.  If $|\widehat\tau-\tau|\le\epsilon_\tau<\Gamma_N/4$, the same proof goes through after replacing the usable transition width by $\Gamma_N-2\epsilon_\tau$.

\section{Bounded polynomial filtering with a finite transition width}
\label{sec:polynomial_filter}
The goal of QG1 is not to determine which of the largest amplitudes is the maximum. Instead, it only needs to suppress amplitudes below the candidate threshold $\lambda$ while preserving amplitudes safely above $\lambda+\Gamma_N$. Amplitudes in the interval $[\lambda,\lambda+\Gamma_N]$ are deliberately left unresolved and remain candidates for the later certification stage.

This finite transition width is essential to the complexity of the filter. Asking a polynomial to reproduce an exact discontinuous step, or to separate two arbitrarily close ranked amplitudes, would make its degree depend on their potentially very small spacing. Instead, the algorithm chooses the finite width $\Gamma_N=(\log_2N)^{-3/2}$ and requires accuracy only outside this transition interval. The purpose of this section is to prove that a bounded polynomial with these properties exists and to determine its degree. We first prove the result for a generic width $\Gamma>0$ and then set $\Gamma=\Gamma_N$ for the quantum extrema-search algorithm. The one-sided construction is sufficient for the proof, while the two-sided magnitude filter required for $|u_i|$ follows immediately afterward.

\begin{lemma}[Bounded finite-width filter polynomial]
Let $-1<\lambda<\lambda+\Gamma<1$ and $0<\eta<1/2$. There exists a real polynomial $p_d$ satisfying
\begin{align}
|p_d(x)| &\leq \eta, && x\in[-1,\lambda], \label{eq:D1}\\
|1-p_d(x)| &\leq \eta, && x\in[\lambda+\Gamma,1], \label{eq:D2}\\
0\leq p_d(x) &\leq 1, && x\in[-1,1], \label{eq:D3}
\end{align}
with degree
\begin{equation}
d=O\!\left(\Gamma^{-1}\log\frac{1}{\eta}\right).
\label{eq:D4}
\end{equation}
\label{lem:gapped_step}
\end{lemma}

\begin{proof}
Define the continuous ramp
\begin{equation}
f_{\lambda,\Gamma}(x)=
\begin{cases}
-1, & x\le\lambda,\\[2pt]
2(x-\lambda)/\Gamma-1, & \lambda<x<\lambda+\Gamma,\\[2pt]
1, & x\ge\lambda+\Gamma.
\end{cases}
\label{eq:D5_ramp}
\end{equation}

The ramp $f_{\lambda,\Gamma}$ has largest slope magnitude $2/\Gamma$ and is therefore Lipschitz with constant $2/\Gamma$. Equivalently, its modulus of continuity satisfies
\begin{equation}
\omega(f_{\lambda,\Gamma},t):=\sup_{|x-y|\leq t}|f_{\lambda,\Gamma}(x)-f_{\lambda,\Gamma}(y)|\leq \frac{2t}{\Gamma}.
\label{eq:D5_modulus}
\end{equation}
Jackson's approximation theorem~\cite{Jackson1911} then guarantees a degree-$m$ polynomial $q_m$ satisfying
\begin{equation}
\|q_m-f_{\lambda,\Gamma}\|_\infty\leq C_J\,\omega(f_{\lambda,\Gamma},m^{-1}),
\label{eq:D6}
\end{equation}
where $C_J$ is an absolute constant. Choosing $m\geq 8C_J/\Gamma$ gives $\|q_m-f_{\lambda,\Gamma}\|_\infty\leq1/4$. Hence $q_m\geq3/4$ in the pass region, $q_m\leq-3/4$ in the reject region, and $|q_m|\leq5/4$ on $[-1,1]$. Rescaling as $\widetilde q_m=(4/5)q_m$ gives
\begin{equation}
\begin{aligned}
|\widetilde q_m(x)|&\leq1,\\
\widetilde q_m(x)&\geq\frac35\quad \text{in the pass region},\\
\widetilde q_m(x)&\leq-\frac35\quad \text{in the reject region}.
\end{aligned}
\label{eq:D8}
\end{equation}

This first stage resolves the transition width $\Gamma$, but only to a fixed constant accuracy. The remaining task is to reduce the error from this constant value to the desired $\eta$ without paying another factor of $1/\Gamma$. We do this by standard majority-vote polynomial amplification~\cite{Buhrman2007}. For $y\in[-1,1]$, define $p=(1+y)/2$. One may regard $p$ as the probability that a Bernoulli trial returns the positive outcome. For an odd integer $r$, so that a majority vote cannot end in a tie, define
\begin{equation}
M_r(y)=2\sum_{k=(r+1)/2}^{r}\binom{r}{k}\left(\frac{1+y}{2}\right)^k\left(\frac{1-y}{2}\right)^{r-k}-1.
\label{eq:D9}
\end{equation}
The sum is the probability that the positive outcome occurs in a strict majority of $r$ independent trials. Thus $M_r(y)$ is simply this majority probability mapped from $[0,1]$ to $[-1,1]$. It is a degree-$r$ polynomial and satisfies $|M_r(y)|\leq1$ for $y\in[-1,1]$.

If $y\geq3/5$, then $p=(1+y)/2\geq4/5$. The majority vote is incorrect only if at most half of the $r$ trials are positive. Hoeffding's inequality~\cite{Hoeffding1963} therefore gives
\begin{equation}
1-M_r(y)\leq2\exp\!\left[-2r\left(\frac45-\frac12\right)^2\right]=2e^{-9r/50}.
\label{eq:D10}
\end{equation}
By symmetry, $1+M_r(y)\leq2e^{-9r/50}$ when $y\leq-3/5$. Hence choosing
\begin{equation}
r\geq\frac{50}{9}\log\frac{2}{\eta}
\label{eq:D11}
\end{equation}
reduces the error in both regions to at most $\eta$.

We now compose the two stages and return from the $\{-1,+1\}$ convention to the desired $\{0,1\}$ filter,
\begin{equation}
p_d(x)=\frac{1+M_r(\widetilde q_m(x))}{2}.
\label{eq:D12}
\end{equation}
Because both $\widetilde q_m$ and $M_r$ are bounded on the required intervals, $0\leq p_d(x)\leq1$ on $[-1,1]$. Moreover, $p_d(x)\leq\eta$ in the reject region and $1-p_d(x)\leq\eta$ in the pass region. Since the degree of a composition is the product of the two degrees,
\begin{equation}
d=mr=O\!\left(\Gamma^{-1}\log\frac{1}{\eta}\right).
\label{eq:D13}
\end{equation}
\label{lem:Lem1 width of step}
\end{proof}
The two factors in the above equation have distinct meanings. The factor $\Gamma^{-1}$ is the cost of resolving the finite transition width, obtained from Jackson approximation, while $\log(1/\eta)$ is the cost of amplifying a constant pass-reject separation to the desired filtering accuracy. Once the two regions are separated by a constant margin, majority amplification drives the error down exponentially. For the quantum extrema-search construction, setting $\Gamma=\Gamma_N=(\log_2N)^{-3/2}$ gives $d=O[(\log_2N)^{3/2}\log(1/\eta_{\rm filt})]$. The QG1 degree is therefore controlled by the deliberately chosen transition width and not by the spacing between top-two or adjacent top-$K$ amplitudes.\\

\ssb{\textbf{Optimality of the transition-width dependence.}
The inverse dependence on the transition width is not an artifact
of the construction above. Consider any real degree-$d$ polynomial
$p_d$ with $\|p_d\|_{\infty,[-1,1]}\leq1$ satisfying
$p_d(\lambda)\leq\eta$ and
$p_d(\lambda+\Gamma)\geq1-\eta$, where $0<\eta<1/2$.
Assume that the transition interval remains a fixed distance
$\rho>0$ from the endpoints of $[-1,1]$.}

\ssb{By the mean-value theorem, there exists
$\xi\in(\lambda,\lambda+\Gamma)$ for which
$|p_d'(\xi)|\geq(1-2\eta)/\Gamma$.
Bernstein's inequality for algebraic polynomials bounded by one
on $[-1,1]$ gives
$|p_d'(x)|\leq d/\sqrt{1-x^2}$.
Since $|\xi|\leq1-\rho$, it follows that
\begin{equation}
d
\geq
\frac{(1-2\eta)\sqrt{\rho(2-\rho)}}{\Gamma}
=
\Omega(\Gamma^{-1}).
\end{equation}
Hence, at fixed filtering accuracy and away from the spectral
endpoints, Lemma~1 attains the optimal scaling
$d=\Theta(\Gamma^{-1})$. The additional
$\log(1/\eta)$ factor in the upper bound controls the requested
high-accuracy pass--reject separation.}

\begin{corollary}[Two-sided magnitude filter]
Under the hypotheses of Lemma 1, with $0\le\lambda<\lambda+\Gamma<1$, there exists an even polynomial $p_{\rm abs}$ of degree $O(\Gamma^{-1}\log\eta^{-1})$ such that $|p_{\rm abs}(x)|\le O(\eta)$ for $|x|\le\lambda$ and $|1-p_{\rm abs}(x)|\le O(\eta)$ for $|x|\ge\lambda+\Gamma$.
\label{cor:two sided filter}
\end{corollary}

\begin{proof}
Let $p_+(x)$ be the one-sided polynomial of Lemma 1 and set $p_-(x)=p_+(-x)$.  Since $0\le p_\pm\le1$, the polynomial
\begin{equation}
p_{\rm abs}(x)
=
1-[1-p_+(x)][1-p_-(x)]
\label{eq:D14}
\end{equation}
is even and bounded between zero and one.  In the central reject region both $p_+$ and $p_-$ are $O(\eta)$, while in either retained tail one of them differs from one by at most $\eta$.  Its degree is at most twice that of $p_+$, which leaves the asymptotic scaling unchanged. \ssb{Applying Lemma~1 with a one-sided accuracy
$\eta=\eta_{\rm filt}/2$ therefore gives
$|p_{\rm abs}(x)|\leq\eta_{\rm filt}$ in the central reject
region and $|1-p_{\rm abs}(x)|\leq\eta_{\rm filt}$ in either
retained tail, without changing the degree scaling.}
\end{proof}

Lemma~\ref{lem:gapped_step} proves the existence and degree of the required bounded finite-width filter.  Circuit cost, however, depends on a second quantity once the polynomial is written in the Chebyshev basis,
\begin{equation}
C_d
=
\sum_{\ell=0}^{d}|c_\ell|.
\label{eq:D15}
\end{equation}
A distinction is important here. The order $d$ controls the largest walk power required in one coherent attempt, whereas $C_d$ controls the normalization of the coefficient-weighted LCU branch.  Degree alone is therefore not the complete implementation cost.

The statement $C_d=O(\log d)$ is not universal for every polynomial satisfying Lemma 1.  It is an additional property of the damped-Chebyshev step expansion used in the numerical implementation.  For the coefficients used here,
\begin{equation}
c_0=\frac{\theta_\lambda}{\pi},
\qquad
c_\ell
=
g_\ell^{(d)}\frac{2}{\pi}\frac{\sin(\ell\theta_\lambda)}{\ell},
\qquad
0\le g_\ell^{(d)}\le1,
\label{eq:D16}
\end{equation}
one has
\begin{equation}
C_d
\le
1+\frac{2}{\pi}\sum_{\ell=1}^{d}\frac1\ell
=O(\log d).
\label{eq:D17}
\end{equation}
The harmonic decay of the coefficients is why hundreds of Chebyshev degrees need not imply a coefficient normalization of order hundreds.  All theorem statements remain valid for other coefficient sets by retaining their actual $C_d$.

\section{Amplitude block encoding and all-degree Chebyshev walk}
\label{sec:chebyshev_walk}
Appendix~\ref{sec:polynomial_filter} constructs a bounded scalar polynomial $p_d(x)=\sum_{\ell=0}^{d}c_\ell T_\ell(x)$ with the desired filtering properties. To apply this polynomial to the amplitudes of $|\Psi\rangle=\sum_i u_i|i\rangle$, we now need a coherent operator whose eigenvalues are the unknown amplitudes $u_i$, followed by a walk whose $\ell$th projected power produces the Chebyshev polynomial $T_\ell(u_i)$. This section constructs these two ingredients.

The first step converts state-preparation access into an amplitude observable without measuring the amplitudes or writing them into a numerical value register. Let $S$ denote the $n$-qubit signal register and $A$ an $n$-qubit auxiliary register. Define the reversible computational-basis copy operation
\begin{equation}
C|i\rangle_S|j\rangle_A=|i\rangle_S|j\oplus i\rangle_A ,
\label{eq:Ecopy}
\end{equation}
so that in particular $C|i\rangle_S|0^n\rangle_A=|i\rangle_S|i\rangle_A$. This operation copies only the computational-basis label and therefore does not copy an unknown quantum state.

\begin{lemma}[Amplitude-observable block encoding]
Let $U_S|0^n\rangle=|\Psi\rangle=\sum_i u_i|i\rangle$, with real amplitudes $u_i$, and define
\begin{equation}
R_\Psi=(I_S\otimes U_S^\dagger)C .
\label{eq:E_Rpsi}
\end{equation}
Then the matrix block of $R_\Psi$ obtained by preparing and projecting the auxiliary register onto $|0^n\rangle_A$ is
\begin{equation}
(I_S\otimes\langle0^n|_A)R_\Psi(I_S\otimes|0^n\rangle_A)=D_\Psi,
\qquad
D_\Psi:=\sum_i u_i|i\rangle\langle i|.
\label{eq:E_block}
\end{equation}
\label{lem:lem2 D_psi}
\end{lemma}

\begin{proof}
For an arbitrary signal state $|\phi\rangle=\sum_i\phi_i|i\rangle$,
\begin{equation}
R_\Psi|\phi\rangle_S|0^n\rangle_A
=
\sum_i\phi_i|i\rangle_S U_S^\dagger|i\rangle_A .
\label{eq:E1}
\end{equation}
Projecting the auxiliary register onto $|0^n\rangle_A$ gives
\begin{equation}
(I_S\otimes\langle0^n|_A)R_\Psi|\phi\rangle_S|0^n\rangle_A
=
\sum_i\phi_i|i\rangle_S\langle0^n|U_S^\dagger|i\rangle .
\label{eq:E2}
\end{equation}
Since $U_S|0^n\rangle=\sum_i u_i|i\rangle$ and the amplitudes are real, $\langle0^n|U_S^\dagger|i\rangle=u_i$. Hence
\begin{equation}
(I_S\otimes\langle0^n|_A)R_\Psi|\phi\rangle_S|0^n\rangle_A
=
\sum_i\phi_i u_i|i\rangle
=
D_\Psi|\phi\rangle .
\label{eq:E3}
\end{equation}
Because this holds for every $|\phi\rangle$, ~\eqref{eq:E_block} follows.
\end{proof}

Therefore, the unknown signal values now appear as eigenvalues, $D_\Psi|i\rangle=u_i|i\rangle$, even though no $u_i$ has been measured or explicitly stored. However, simply applying powers of $R_\Psi$ does not produce powers of $D_\Psi$. If $R_\Psi$ is written schematically in blocks as $R_\Psi=\left(\begin{smallmatrix}D_\Psi&B\\ C&E\end{smallmatrix}\right)$, then the corresponding block of $R_\Psi^2$ is $D_\Psi^2+BC$, because amplitude can leave the selected auxiliary subspace and later return. We therefore construct a walk whose projected powers obey an exact polynomial identity.

The unitary $R_\Psi$ is not generally Hermitian. We first embed it into the larger Hermitian unitary
\begin{equation}
\mathcal U_\Psi
=
|0\rangle\langle1|_q\otimes R_\Psi
+
|1\rangle\langle0|_q\otimes R_\Psi^\dagger ,
\label{eq:E_hermitianization}
\end{equation}
where $q$ is one additional qubit. Define the projector
\begin{equation}
\Pi
=
|+\rangle\langle+|_q
\otimes
|0^n\rangle\langle0^n|_A
\otimes I_S ,
\label{eq:E_projector}
\end{equation}
onto the auxiliary state $|+\rangle_q|0^n\rangle_A$. We refer to this selected subspace as the \emph{good subspace}. The reflection about it is $Z_\Pi:=2\Pi-I$, and the Chebyshev walk is
\begin{equation}
W_\Psi:=Z_\Pi\mathcal U_\Psi .
\label{eq:E_walk}
\end{equation}

\begin{lemma}[All-degree Chebyshev walk]
The walk $W_\Psi$ defined above satisfies
\begin{equation}
\Pi W_\Psi^\ell\Pi
=
T_\ell(D_\Psi)\Pi
\label{eq:E4}
\end{equation}
for every integer $\ell\geq0$.
\label{lem:lem3 all degree walk}
\end{lemma}

\begin{proof}
From ~\eqref{eq:E_hermitianization},
\begin{equation}
\mathcal U_\Psi^\dagger=\mathcal U_\Psi,
\qquad
\mathcal U_\Psi^2=I.
\label{eq:E5}
\end{equation}
Thus $\mathcal U_\Psi$ is both unitary and Hermitian. Projecting it onto the good auxiliary state gives
\begin{equation}
\begin{aligned}
&(\langle+|_q\langle0^n|_A\otimes I_S)
\mathcal U_\Psi
(|+\rangle_q|0^n\rangle_A\otimes I_S)\\
&\qquad=
\frac{1}{2}(D_\Psi+D_\Psi^\dagger)
=
D_\Psi .
\end{aligned}
\label{eq:E6}
\end{equation}
where the final equality follows because $D_\Psi$ is real diagonal. Hermitianization therefore changes the surrounding unitary while preserving the desired amplitude observable in the selected block.

We now examine one amplitude at a time. Fix an index $i$ and define the corresponding state in the good subspace as
\begin{equation}
|g_i\rangle:=|+\rangle_q|0^n\rangle_A|i\rangle_S .
\label{eq:E_good}
\end{equation}
Since $D_\Psi|i\rangle=u_i|i\rangle$, ~\eqref{eq:E6} implies $\langle g_i|\mathcal U_\Psi|g_i\rangle=u_i$. Define
\begin{equation}
s_i:=\sqrt{1-u_i^2},
\qquad
|b_i\rangle:=\frac{\mathcal U_\Psi|g_i\rangle-u_i|g_i\rangle}{s_i},
\label{eq:E_bad}
\end{equation}
when $s_i\neq0$. The state $|b_i\rangle$ is normalized, orthogonal to $|g_i\rangle$, and lies outside the good subspace. The two states therefore span the two-dimensional subspace explored by the walk for the amplitude $u_i$.

Using $\mathcal U_\Psi^2=I$ gives
\begin{equation}
\mathcal U_\Psi|b_i\rangle
=
s_i|g_i\rangle-u_i|b_i\rangle .
\label{eq:E10}
\end{equation}
Hence, in the ordered basis $\{|g_i\rangle,|b_i\rangle\}$,
\begin{equation}
\begin{aligned}
\mathcal U_\Psi&=
\begin{pmatrix}
u_i&s_i\\
s_i&-u_i
\end{pmatrix},
&
Z_\Pi&=
\begin{pmatrix}
1&0\\
0&-1
\end{pmatrix},\\[4pt]
W_\Psi&=
\begin{pmatrix}
u_i&s_i\\
-s_i&u_i
\end{pmatrix}.
\end{aligned}
\label{eq:E11}
\end{equation}

Because $|u_i|\leq1$, write $u_i=\cos\theta_i$ and $s_i=\sin\theta_i$, with $\theta_i=\arccos u_i$. The walk is therefore a rotation by $-\theta_i$ in this two-dimensional subspace, and after $\ell$ steps
\begin{equation}
W_\Psi^\ell=
\begin{pmatrix}
\cos(\ell\theta_i)&\sin(\ell\theta_i)\\
-\sin(\ell\theta_i)&\cos(\ell\theta_i)
\end{pmatrix}.
\label{eq:E13}
\end{equation}
Projecting back onto the good subspace selects the upper-left matrix element,
\begin{equation}
\langle g_i|W_\Psi^\ell|g_i\rangle
=
\cos(\ell\theta_i)
=
T_\ell(\cos\theta_i)
=
T_\ell(u_i).
\label{eq:E14}
\end{equation}
The limiting cases $u_i=\pm1$, for which $s_i=0$, follow directly by continuity. Since the basis states $|i\rangle$ span the signal register, ~\eqref{eq:E4} follows.
\end{proof}

\eqref{eq:E14} explains the origin of the Chebyshev transformation. Each signal amplitude determines a rotation angle through $u_i=\cos\theta_i$. Repeating the walk $\ell$ times multiplies this angle by $\ell$, and projection onto the good subspace returns $\cos(\ell\theta_i)=T_\ell(u_i)$. Thus the walk converts the hidden amplitude eigenvalues into the complete Chebyshev basis $T_0(D_\Psi),T_1(D_\Psi),\ldots,T_d(D_\Psi)$ required by the polynomial constructed in Appendix~\ref{sec:polynomial_filter}. The coefficient-weighted combination of these basis operators into $p_d(D_\Psi)$ is carried out in Appendix~\ref{sec:theorem_proof}.

\section{Coefficient synthesis, success probability, and proof of Theorem 1}
\label{sec:theorem_proof}

Appendix~\ref{sec:polynomial_filter} constructs the bounded finite-width polynomial $p_d(x)=\sum_{\ell=0}^{d}c_\ell T_\ell(x)$, while Appendix~\ref{sec:chebyshev_walk} constructs the walk satisfying $\Pi W_\Psi^\ell\Pi=T_\ell(D_\Psi)\Pi$. It remains to combine these walk powers with the coefficients $c_\ell$, determine the probability of obtaining the transformed state, and combine this cost with candidate recovery and final certification. Throughout this section we write $L:=\log_2N$ and $C_d:=\sum_{\ell=0}^{d}|c_\ell|$.

\subsection{Coefficient-weighted synthesis}
\label{subsec:coefficient_synthesis}

Introduce a history register $H$ with $\lceil\log_2(d+1)\rceil$ qubits whose basis state $|\ell\rangle_H$ labels the Chebyshev degree $\ell$. Let $P_c$ prepare the normalized coefficient state
\begin{equation}
P_c|0\rangle_H=\frac{1}{\sqrt{C_d}}\sum_{\ell=0}^{d}\sqrt{|c_\ell|}\,|\ell\rangle_H,
\label{eq:F1}
\end{equation}
and let $S_c$ apply the sign of each nonzero coefficient, $S_c|\ell\rangle_H=\operatorname{sgn}(c_\ell)|\ell\rangle_H$. Starting from $|0\rangle_H|+\rangle_q|0^n\rangle_A|\Psi\rangle_S$, these operations prepare
\begin{equation}
\frac{1}{\sqrt{C_d}}\sum_{\ell=0}^{d}\operatorname{sgn}(c_\ell)\sqrt{|c_\ell|}\,|\ell\rangle_H|+\rangle_q|0^n\rangle_A|\Psi\rangle_S .
\label{eq:F2}
\end{equation}
The history value then controls the corresponding walk power through
\begin{equation}
\operatorname{SELECT}(W_\Psi):=\sum_{\ell=0}^{d}|\ell\rangle\langle\ell|_H\otimes W_\Psi^\ell ,
\label{eq:F3}
\end{equation}
giving
\begin{equation}
\frac{1}{\sqrt{C_d}}\sum_{\ell=0}^{d}\operatorname{sgn}(c_\ell)\sqrt{|c_\ell|}\,|\ell\rangle_H W_\Psi^\ell|+\rangle_q|0^n\rangle_A|\Psi\rangle_S .
\label{eq:F4}
\end{equation}
Since $\langle0|_H P_c^\dagger|\ell\rangle_H=\sqrt{|c_\ell|/C_d}$, unpreparing the history register and projecting it onto $|0\rangle_H$ supplies the second factor $\sqrt{|c_\ell|/C_d}$. The product of the two factors, together with the sign phase, is therefore $c_\ell/C_d$. The resulting branch is
\begin{equation}
\frac{1}{C_d}\sum_{\ell=0}^{d}c_\ell W_\Psi^\ell|+\rangle_q|0^n\rangle_A|\Psi\rangle_S .
\label{eq:F5}
\end{equation}
Finally projecting the walk ancillas onto the good subspace $\Pi$ and using the result of Appendix~\ref{sec:chebyshev_walk} gives
\begin{equation}
\frac{1}{C_d}\sum_{\ell=0}^{d}c_\ell T_\ell(D_\Psi)|\Psi\rangle=\frac{1}{C_d}p_d(D_\Psi)|\Psi\rangle .
\label{eq:F6}
\end{equation}
Thus the coefficient-weighted walk realizes exactly the nonlinear amplitude transformation stated in Lemma 1.

The coefficient normalization $C_d$ determines the amplitude penalty introduced by this coherent sum. A useful bound follows directly from the boundedness of the polynomial. Writing $x=\cos\theta$ gives $p_d(\cos\theta)=c_0+\sum_{\ell=1}^{d}c_\ell\cos(\ell\theta)$. Since $|p_d(x)|\leq1$ on $[-1,1]$, orthogonality of the cosine functions gives
\begin{equation}
c_0^2+\frac{1}{2}\sum_{\ell=1}^{d}c_\ell^2=\frac{1}{\pi}\int_0^\pi |p_d(\cos\theta)|^2\,d\theta\leq1.
\label{eq:F7}
\end{equation}
Therefore, by the Cauchy--Schwarz inequality,
\begin{equation}
C_d=\sum_{\ell=0}^{d}|c_\ell|\leq1+\sqrt{2d}=O(\sqrt d).
\label{eq:F8}
\end{equation}
Since $d$ is polylogarithmic in $N$, this bound alone guarantees that $C_d$ is also polylogarithmic. Particular damped Chebyshev realizations can have the stronger $C_d=O(\log d)$ behavior used in the numerical implementation.

\subsection*{B. Retained mass and success probability}

\ssb{Define
\begin{equation}
P_e
:=
\bigl\|p_d(D_\Psi)|\Psi\rangle\bigr\|_2^2
=
\sum_i |p_d(u_i)|^2|u_i|^2 .
\end{equation}
Equation~(47) therefore gives
$p_{\rm succ}=P_e/C_d^2$.

Let $P_R:=\sum_{i\in R_\Gamma}|u_i|^2$. For every
$i\in R_\Gamma$, the filter satisfies
$|p_d(u_i)|\geq1-\eta_{\rm filt}$, and hence
\begin{equation}
P_e
\geq
(1-\eta_{\rm filt})^2P_R.
\end{equation}}

\ssb{The fraction of the normalized filtered state lying in the
candidate set is
$q_\Gamma:=
P_e^{-1}\sum_{i\in C_\Gamma}|p_d(u_i)|^2|u_i|^2$.
Since $|p_d(u_i)|\leq\eta_{\rm filt}$ outside $C_\Gamma$,
\begin{equation}
1-q_\Gamma
\leq
\frac{\eta_{\rm filt}^2}{P_e}
\leq
\frac{\eta_{\rm filt}^2}
{(1-\eta_{\rm filt})^2P_R}.
\end{equation}
Thus $P_R$ controls the cost of obtaining the retained branch,
while $\eta_{\rm filt}^2/P_R$ controls the normalized leakage
outside the candidate set.}

\subsection*{C. Cost of preparing a filtered candidate state}

\ssb{One implementation of $\mathrm{SELECT}(W_\Psi)$ requires at
most $d$ elementary applications of the Chebyshev walk. Without
amplitude amplification,
\begin{equation}
Q_{\rm QG1}
=
O\!\left[
\frac{dC_d^2}{P_e}
\log\frac{1}{\delta_{\rm post}}
\right].
\end{equation}
With amplitude amplification,
\begin{equation}
\begin{aligned}
Q_{\rm QG1}^{\rm AA}
&=
O\!\left[
\frac{dC_d}{\sqrt{P_e}}
\log\frac{1}{\delta_{\rm post}}
\right]\\
&=
O\!\left[
\frac{dC_d}
{(1-\eta_{\rm filt})\sqrt{P_R}}
\log\frac{1}{\delta_{\rm post}}
\right].
\end{aligned}
\end{equation}}

\ssb{The square-root dependence on the successful-branch probability
is optimal in the black-box amplification model. Since
$p_{\rm succ}=P_e/C_d^2$, amplification to constant success
probability requires $\Omega(C_d/\sqrt{P_e})$ coherent uses in
general. The amplitude-amplified construction therefore attains
the optimal dependence on the postselection amplitude, up to the
confidence factor.}

\subsection*{D. Proof of Theorem 1 and polylogarithmic specialization}

\ssb{Appendix~\ref{sec:polynomial_filter} gives
$d=O[\Gamma^{-1}\log(1/\eta_{\rm filt})]$, while
Appendix~\ref{sec:chebyshev_walk} and Sec.~\ref{subsec:coefficient_synthesis} realize the coherent branch
$C_d^{-1}p_d(D_\Psi)|\Psi\rangle$. The preceding bounds give
$P_e\geq(1-\eta_{\rm filt})^2P_R$ and hence the per-state
preparation cost stated in Theorem~1.}

\ssb{Let $K:=|C_\Gamma|$. Appendix~\ref{app:global_completeness} proves that
\begin{equation}
M_{\rm glob}
=
O\!\left[
\frac{K}{(1-\eta_{\rm filt})^2}
\left(
1+\frac{\eta_{\rm filt}^2}{P_R}
\right)
\log\frac{1}{\delta_{\rm glob}}
\right]
\end{equation}
successful filtered states suffice to recover a global-extremum
address.}

\ssb{Multiplying this by the cost of one successful filtered state
and adding threshold calibration and final candidate
certification gives
\begin{multline}
Q_{\rm ext}^{\rm AA}
=
Q_{\rm cal}
+
\widetilde O\!\Biggl[
K\left(1+\frac{\eta_{\rm filt}^2}{P_R}\right)\\
\times\Biggl(
\frac{dC_d}{(1-\eta_{\rm filt})^3\sqrt{P_R}}\\
+
\frac{1}{(1-\eta_{\rm filt})^2\epsilon_{\rm QAE}}
\Biggr)
\Biggr].
\end{multline}
This proves the general complexity statement of Theorem~1.}

\ssb{For the structured family used below, let
$L:=\log_2N$, $\Gamma_N:=L^{-3/2}$, and
$P_R\geq c_PL^{-a}$ for fixed $a>0$.
Choosing
$\lambda\geq(\kappa_{\rm cand}L)^{-1/2}$ implies
$K\leq\kappa_{\rm cand}L$ by normalization.
If additionally
$\eta_{\rm filt}^2\leq\xi P_R$ with $0<\xi<1$, then
$M_{\rm glob}=O[L\log(1/\delta_{\rm glob})]$ and
\begin{equation}
Q_{\rm ext}^{\rm AA}
=
Q_{\rm cal}
+
\widetilde O\!\left[
C_dL^{(a+5)/2}
+
\frac{L}{\epsilon_{\rm QAE}}
\right].
\end{equation}
The numerical study uses the representative case $a=1$.}

\ssb{Finally, choosing the calibration, postselection,
global-recovery, and certification failure probabilities so that
$\delta_{\rm cal}+\delta_{\rm post}+\delta_{\rm glob}
+\delta_{\rm cert}\leq\delta$ gives the stated overall failure
probability by the union bound.}

\section{Quantum Gadget 2: local contrast, fuzzy search, and greedy sort}
\label{app:qg2}

QG2 contains two independent optional refinements. QG2a applies a local-excess transformation before QG1 so that amplitudes exceptional relative to their neighborhood can become more prominent. QG2b acts after candidate-state preparation and combines a coarse binary search with a greedy sorting of the binary partitions used to expose the extremal address. Neither refinement is required for Theorem~1. The unmodified input and explicit candidate enumeration remain rigorous fallbacks.

\subsection{Local-excess operator}
\label{subsec:qg2_local_excess}

Let $S|i\rangle=|i+1\pmod N\rangle$ denote the cyclic shift. For any positive integer $m$, define
\begin{equation}
B_m
:=
\frac{1}{2m}
\sum_{\ell=1}^{m}
(S^\ell+S^{-\ell}),
\qquad
A_m:=I-B_m .
\label{eq:G1}
\end{equation}
The operator $B_m$ forms the average of the $m$ neighboring samples on either side, excluding the point itself while $A_m$ subtracts this local background. The integer $m$ therefore controls the spatial scale of the comparison.

For a constant signal $u_i=c$, one has $B_mu=u$ and hence $A_mu=0$. For a smooth signal sampled with spacing $\Delta x$, Taylor expansion gives
\begin{equation}
\begin{aligned}
u(x+\ell\Delta x)+u(x-\ell\Delta x)
&=
2u(x)+\ell^2\Delta x^2u''(x)\\
&\quad+
O(\ell^4\Delta x^4).
\end{aligned}
\label{eq:G2}
\end{equation}
Substituting this expression into \eqref{eq:G1} and using $\sum_{\ell=1}^{m}\ell^2=m(m+1)(2m+1)/6$ gives
\begin{equation}
(A_mu)(x)
=
-\frac{(m+1)(2m+1)}{12}
\Delta x^2u''(x)
+
O(m^4\Delta x^4).
\label{eq:G3}
\end{equation}
Thus a locally smooth background is suppressed to leading order, while sufficiently localized curvature, peaks, and troughs remain.

The same behavior is visible in Fourier space. For a mode $u_j=e^{ikj}$,
\begin{equation}
B_m(k)
=
\frac{1}{m}
\sum_{\ell=1}^{m}\cos(k\ell),
\qquad
A_m(k)
=
1-\frac{1}{m}
\sum_{\ell=1}^{m}\cos(k\ell).
\label{eq:G4}
\end{equation}
Using $\cos(k\ell)=1-k^2\ell^2/2+O(k^4\ell^4)$ gives
\begin{equation}
A_m(k)
=
\frac{(m+1)(2m+1)}{12}k^2
+
O(m^4k^4).
\label{eq:G5}
\end{equation}
Hence $A_m$ behaves as a local-excess or high-pass transformation at small wave number. There is no universally preferred value of $m$. Different signal structures may favor different neighborhood scales.

A position-dependent local background can change the computational objective. QG2a is therefore optional when the desired quantity is the extremum of the original signal itself. The unmodified preprocessing choice is represented by
\begin{equation}
A_0:=I .
\label{eq:G6}
\end{equation}

\subsection{LCU implementation and combined useful mass}
\label{app:qg2_lcu}

For $m\geq1$,
\begin{equation}
A_m
=
I
-
\frac{1}{2m}\sum_{\ell=1}^{m}S^\ell
-
\frac{1}{2m}\sum_{\ell=1}^{m}S^{-\ell}.
\label{eq:G7}
\end{equation}
The sum of the absolute LCU coefficients is therefore
\begin{equation}
C_{A,m}
=
1+\frac{m}{2m}+\frac{m}{2m}
=
2,
\qquad
m\geq1,
\label{eq:G8}
\end{equation}
while $C_{A,0}=1$.

Define the residual energy
\begin{equation}
E_m
:=
\|A_m|\Psi\rangle\|_2^2 .
\label{eq:G9}
\end{equation}
The unnormalized LCU branch $A_m|\Psi\rangle/C_{A,m}$ occurs with probability $E_m/C_{A,m}^2$. Conditional on this branch, the normalized residual state is
\begin{equation}
|\Psi_m\rangle
=
\frac{A_m|\Psi\rangle}{\sqrt{E_m}} .
\label{eq:G10}
\end{equation}
For $m=0$, $E_0=1$ and $|\Psi_0\rangle=|\Psi\rangle$.

Let
\begin{equation}
D_m
:=
\sum_i
\langle i|\Psi_m\rangle
|i\rangle\langle i|
\label{eq:G11}
\end{equation}
denote the amplitude observable of the normalized residual state, and define the QG1 retained mass
\begin{equation}
F_m
:=
\|p_{d_m}(D_m)|\Psi_m\rangle\|_2^2 .
\label{eq:G12}
\end{equation}
The total useful mass relative to the original input is then
\begin{equation}
P_m
:=
E_mF_m
=
\|p_{d_m}(D_m)A_m|\Psi\rangle\|_2^2 .
\label{eq:G13}
\end{equation}

The distinction between $F_m$ and $P_m$ is essential. Normalization can make a very small residual appear sharply concentrated. If $E_m$ is small, however, obtaining that normalized residual is itself expensive. The product $P_m$ accounts simultaneously for residual-state preparation and subsequent QG1 filtering.

Combining amplitude amplification for the residual preparation and QG1 gives the query cost of one successful candidate state,
\begin{equation}
Q_m^{\rm AA}
=
O\!\left[
\frac{C_{A,m}d_mC_{d_m}}{\sqrt{P_m}}
\log\!\left(\frac{1}{\delta_m}\right)
\right].
\label{eq:G14}
\end{equation}
Here $d_m$ is the QG1 polynomial degree for the normalized residual, $C_{d_m}$ is its Chebyshev coefficient normalization, and $\delta_m$ is the allowed failure probability for preparing the successful candidate state.

The width $m$ does not appear linearly in \eqref{eq:G14}. The shifts $S^\ell$ can be indexed coherently and implemented by modular arithmetic on the $O(\log N)$-qubit address register. We therefore restrict attention to admissible values of $m$ for which the corresponding state preparation, index preparation, and controlled shifts have polylogarithmic gate overhead. Any additional implementation-dependent cost must be included when this condition is not satisfied.

\subsection{Proof of Proposition 2: guarded local-excess preprocessing}
\label{app:qg2_selection}

The neighborhood width is treated as a calibration parameter rather than a fixed constant. Let $\mathcal M$ be a prescribed admissible set satisfying
\begin{equation}
0\in\mathcal M,
\qquad
|\mathcal M|=\polylog{N}
\label{eq:G15}
\end{equation}
at most. A fixed finite family is included as a special case.

For each $m\in\mathcal M$, QG1 uses the same designed transition-width scaling $\Gamma_N=(\log_2N)^{-3/2}$. Hence
\begin{equation}
d_m
=
O\!\left[
(\log_2N)^{3/2}
\log\!\left(
\frac{1}{\eta_{{\rm filt},m}}
\right)
\right].
\label{eq:G16}
\end{equation}
The quantities $C_{A,m}$ and $C_{d_m}$ are known from the chosen constructions, while the branch probabilities entering $P_m$ are estimated during the one-time calibration stage. The predicted production cost \eqref{eq:G14} can therefore be evaluated for each admissible width.

The production width is selected according to
\begin{equation}
m^\star
\in
\arg\min_{m\in\mathcal M}
Q_m^{\rm AA}.
\label{eq:G17}
\end{equation}
Since $0\in\mathcal M$ represents the original unmodified signal, the selected preprocessing cannot be worse than the verified raw-QG1 choice,
\begin{equation}
Q_{m^\star}^{\rm AA}
\leq
Q_0^{\rm AA}.
\label{eq:G80}
\end{equation}
A nonzero local-excess width is therefore accepted only if the increased candidate contrast and useful filtered mass compensate for the residual-state preparation cost.

If the selected width satisfies
\begin{equation}
P_{m^\star}
=
\Omega\!\left(
(\log_2N)^{-1}
\right),
\label{eq:G19}
\end{equation}
then combining \eqref{eq:G14}, \eqref{eq:G16}, and \eqref{eq:G19} gives
\begin{equation}
Q_{m^\star}^{\rm AA}
=
\widetilde O\!\left[
C_{d_{m^\star}}(\log_2N)^2
\right].
\label{eq:G20}
\end{equation}
Testing the admissible widths requires at most $\polylog{N}$ additional calibration work because $|\mathcal M|=\polylog{N}$ at most. This work is included in the one-time calibration term $Q_{\rm cal}$. This proves Proposition~2.

\subsection{Proof of Proposition 3: fuzzy search and greedy sort}
\label{app:qg2_fuzzy_greedy}

QG2b begins after a successful candidate state has been prepared. Its purpose is to expose the address of a dominant candidate without estimating and sorting all candidate values. The procedure has two complementary parts.

The \emph{fuzzy search} traverses the $L:=\log_2N$ levels of the binary address tree. At each level it asks only which of the two child blocks carries more candidate probability. This is a coarse binary decision and does not attempt to resolve the exact values of the amplitudes inside either child. The \emph{greedy sort} records the favored orientation at each level and progressively arranges the corresponding binary partitions toward a terminal address. It therefore sorts only the sequence of partitions needed for localization rather than globally sorting all $N$ amplitudes.

Draw $M$ independent successful candidate states and measure their complete computational-basis addresses. A measurement of one complete address contains all $L$ address bits. The resulting data can therefore be viewed equivalently as supplying the outcome of each of the $L$ coarse prefix measurements required by the fuzzy search.

For $r=0,\ldots,L-1$, let $B_r$ denote the depth-$r$ prefix block on the path containing the true extremum $i^\star$, and define
\begin{equation}
P_r
:=
\Pr[i\in B_r].
\label{eq:G21}
\end{equation}
Let $B_{r+1}^\star$ denote the child of $B_r$ that contains $i^\star$. Conditioned on a sample belonging to $B_r$, define the child bias $\zeta_r$ through
\begin{equation}
\Pr[
i\in B_{r+1}^\star
\mid
i\in B_r
]
=
\frac{1+\zeta_r}{2}.
\label{eq:G22}
\end{equation}
Thus $\zeta_r=0$ corresponds to equally weighted children, while $\zeta_r>0$ means that the child containing the true extremum is favored. Proposition~3 assumes
\begin{equation}
P_r\geq p_0,
\qquad
\zeta_r\geq\zeta_0,
\qquad
r=0,\ldots,L-1 .
\label{eq:G23}
\end{equation}

Let $N_r$ denote the number of the $M$ measured addresses that fall inside $B_r$. Since $\mathbb E[N_r]=MP_r$, a multiplicative Chernoff bound gives
\begin{equation}
\Pr\!\left[
N_r<\frac{MP_r}{2}
\right]
\leq
\exp\!\left(
-\frac{MP_r}{8}
\right)
\leq
\exp\!\left(
-\frac{Mp_0}{8}
\right).
\label{eq:G24}
\end{equation}
Hence, with high probability, sufficiently many of the original measurements remain available when the fuzzy search reaches any prefix on the correct path.

Conditioned on $N_r\geq MP_r/2$, the $N_r$ samples inside $B_r$ form binary observations of its two child blocks. The child containing $i^\star$ has probability at least $(1+\zeta_0)/2$. Hoeffding's inequality therefore gives
\begin{equation}
\begin{aligned}
&\Pr[\text{wrong child selected at level }r]\\
&\qquad\leq
\exp\!\left(-\frac{N_r\zeta_0^2}{2}\right)\\
&\qquad\leq
\exp\!\left(-\frac{Mp_0\zeta_0^2}{4}\right).
\end{aligned}
\label{eq:G25}
\end{equation}

At every level $r$, the statistically favored child is moved to a fixed branch by a reversible controlled swap $G_r$. After $L=\log_2N$ levels, the accumulated transformation is
\begin{equation}
G_{\rm gr}
=
G_{L-1}\cdots G_1G_0 .
\label{eq:G26}
\end{equation}
The fuzzy search is therefore carried out in a progressively reordered address space. Let $j^\star$ denote the terminal address obtained after the last binary decision. Because all sorting operations are reversible permutations, the original address of the candidate is recovered by preparing the basis state $|j^\star\rangle$ and applying the inverse sequence,
\begin{equation}
G_{\rm gr}^{\dagger}|j^\star\rangle
=
G_0^{\dagger}G_1^{\dagger}\cdots G_{L-1}^{\dagger}|j^\star\rangle
=
|i^\star\rangle .
\label{eq:G27_retrace}
\end{equation}
Thus no information about the original position is lost during the greedy sorting. The sorting circuit is used only to simplify localization, and its inverse exactly restores the physical index after the extremal location has been identified. The sorting operations therefore require only the recorded binary orientations and do not require numerical sorting of all signal amplitudes.

Choosing
\begin{equation}
M
\geq
\frac{c_{\rm fg}}{p_0\zeta_0^2}
\log\!\left(
\frac{2L}{\delta_{\rm fg}}
\right)
\label{eq:G27}
\end{equation}
with a sufficiently large absolute constant makes the combined insufficient-count and wrong-child failure probability at each level at most $\delta_{\rm fg}/L$. A union bound over the $L$ fuzzy-search decisions then gives total failure probability at most $\delta_{\rm fg}$.

An important point is that \eqref{eq:G27} is the total number of successful candidate-state measurements and not a separate sample count at each level. Once a complete address has been measured, the same address can be tested against every prefix $B_r$. Thus the same collection of $M$ addresses is reused throughout the binary tree. The depth $L=\log_2N$ therefore contributes only through the confidence factor $\log(L/\delta_{\rm fg})$, rather than as an additional multiplicative factor of $L$ in the number of quantum-state preparations.

Writing the required number of successful candidate states as $M_{\rm fg}$, the total query complexity of QG2b is
\begin{equation}
Q_{\rm fg}^{\rm AA}
=
M_{\rm fg}
Q_{m^\star}^{\rm AA}.
\label{eq:G28}
\end{equation}
For fixed hierarchy parameters $p_0$ and $\zeta_0$,
\begin{equation}
M_{\rm fg}
=
O\!\left[
\log\!\left(
\frac{\log_2N}{\delta_{\rm fg}}
\right)
\right].
\label{eq:G29}
\end{equation}
This proves Proposition~3.

The hierarchy condition in \eqref{eq:G23} is an additional property of the filtered candidate distribution and is not required by Theorem~1 of the Main Text. The result above should therefore be interpreted as a statistical guarantee for an optional localization and address-retracing procedure, not as an independent asymptotic quantum-speedup claim. If the hierarchy condition cannot be established, QG2b is not invoked and the explicit candidate-enumeration procedure of Appendix~\ref{sec:Candidates,enum,cert,sign,cases} remains the rigorous fallback.

\section{Candidate enumeration, certification, sign estimation and general extremum cases}
\label{sec:Candidates,enum,cert,sign,cases}

QG1 deliberately stops at a candidate isolation stage.  That state is already compressed onto a short tail, but the short tail sample has not yet delivered a certified extremum. Therefore one must recover the surviving addresses, compare their values at sufficient precision, and for signed objectives, recover the signs.  This appendix closes that final loop and also distinguishes two different notions of ``minimum'' that should not be conflated -- the most negative value and the smallest magnitude.

\subsection{Enumeration fallback}
\label{subsec:enuemration fallback}
Let $C$ be the active candidate set in the normalized filtered distribution and suppose
\begin{equation}
|C|=K\le\kappa_{\rm cand}L,
\qquad
q_i\ge q_{\min}>0
\quad\text{for every }i\in C.
\label{eq:H1}
\end{equation}
After $M$ successful filtered-state measurements, candidate $i$ is missed with probability
\begin{equation}
(1-q_i)^M\le e^{-Mq_{\min}}.
\label{eq:H2}
\end{equation}
A union bound over $K$ candidates gives
\begin{equation}
\Pr[\text{some candidate missed}]
\le
Ke^{-Mq_{\min}}.
\label{eq:H3}
\end{equation}
Thus
\begin{equation}
M_{\rm enum}
\ge
\frac1{q_{\min}}\log\frac{K}{\delta_{\rm enum}}
\label{eq:H4}
\end{equation}
is sufficient.  If $q_{\min}=\Omega(L^{-1})$, then $M_{\rm enum}=O[L\log(L/\delta_{\rm enum})]$.  This is the rigorous fallback when the stronger binary-hierarchy promise needed by fuzzy-greedy localization is unavailable.

\subsection{Value certification}
\label{subsec:value certification}

For a candidate $i$, the projector $\Pi_i=\ket{i}\!\bra{i}$ has success probability
\begin{equation}
p_i=|u_i|^2.
\label{eq:H5}
\end{equation}
Iterative amplitude estimation estimates $p_i$ to additive error $\epsilon$ with
\begin{equation}
O\!\left(\epsilon^{-1}\log\delta_i^{-1}\right)
\label{eq:H6}
\end{equation}
controlled uses of $U_S$ and $U_S^\dagger$ \cite{brassard2002quantum,grinko2021iterative}.  If
\begin{equation}
\delta_p=p_{(1)}-p_{(2)}>0,
\label{eq:H7}
\end{equation}
choosing $\epsilon=\delta_p/4$ preserves the ordering of the two largest probabilities.  Assigning $\delta_i=\delta_{\rm cert}/K$ and using a union bound gives
\begin{equation}
Q_{\rm cert}^{\rm QAE}
=
O\!\left[
\frac{K}{\delta_p}
\log\!\left(\frac{K}{\delta_{\rm cert}}\right)
\right].
\label{eq:H8}
\end{equation}
If a calibrated sensor or value oracle returns a dimensionless physical value $\bar g_i=g_i/G$, replace $\delta_p$ by
\begin{equation}
\delta_g
=
|\bar g_{(1)}-\bar g_{(2)}|.
\label{eq:H9}
\end{equation}
Exact digital values require only $K-1=O(L)$ comparisons.  The decisive separation therefore enters only after candidate compression; it is not the transition width used to construct QG1.

\subsection{Sign estimation after candidate localization}
\label{subsec:sign estimation}
Computational-basis measurement and amplitude estimation of the singleton projector $\Pi_i=\ket{i}\!\bra{i}$ return $|u_i|^2$ and therefore do not reveal the sign of a real amplitude.  The principal access model does not assume a sign oracle.  Sign information is recovered only after QG1/QG2 have reduced the search to $K=O(L)$ candidates.

For a known candidate index $i$, let $U_i\ket{0^n}=\ket{i}$ be the basis-state loader, implemented by at most $n$ Pauli-$X$ gates.  Using controlled $U_i$ and controlled $U_S$, prepare
\begin{equation}
\ket{\Omega_i}
=
\frac1{\sqrt2}\left(
\ket{0}\ket{i}+\ket{1}\ket{\Psi}
\right).
\label{eq:H40}
\end{equation}
After a Hadamard gate on the ancilla,
\begin{equation}
\ket{\Omega_i'}
=
\frac12\ket{0}(\ket{i}+\ket{\Psi})
+
\frac12\ket{1}(\ket{i}-\ket{\Psi}).
\label{eq:H41}
\end{equation}
Since the amplitudes are real,
\begin{equation}
p_{i,0}
:=
\Pr[a=0]
=
\frac14\norm{\ket{i}+\ket{\Psi}}_2^2
=
\frac{1+u_i}{2},
\label{eq:H42}
\end{equation}
\begin{equation}
p_{i,1}
:=
\Pr[a=1]
=
\frac{1-u_i}{2}.
\label{eq:H43}
\end{equation}
Therefore
\begin{equation}
u_i
=
p_{i,0}-p_{i,1}
=
2p_{i,0}-1.
\label{eq:H44}
\end{equation}
Direct sampling estimates $u_i$ to additive error $\epsilon_{\rm sgn}$ with
\begin{equation}
O\!\left(
\epsilon_{\rm sgn}^{-2}\log\delta_i^{-1}
\right)
\label{eq:H45}
\end{equation}
preparations, whereas iterative amplitude estimation gives
\begin{equation}
O\!\left(
\epsilon_{\rm sgn}^{-1}\log\delta_i^{-1}
\right)
\label{eq:H46}
\end{equation}
controlled preparations.  Assigning $\delta_i=\delta_{\rm sgn}/K$ and applying a union bound gives
\begin{equation}
Q_{\rm sign}
=
O\!\left[
\frac{K}{\epsilon_{\rm sgn}}
\log\!\left(\frac{K}{\delta_{\rm sgn}}\right)
\right]
\label{eq:H47}
\end{equation}
with amplitude estimation.  Hence $K=O(L)$ and $\epsilon_{\rm sgn}^{-1}=\polylog(N)$ preserve polylogarithmic scaling.  To decide only the sign of a candidate promised to satisfy $|u_i|\ge\gamma_{\rm sgn}$, it is sufficient to choose $\epsilon_{\rm sgn}<\gamma_{\rm sgn}/2$.  To rank signed candidates separated by $\delta_u$, it is sufficient to choose $\epsilon_{\rm sgn}<\delta_u/4$.

The interference experiment fixes a phase convention through the controlled preparation circuit $U_S$.  With copies of an unknown state alone, $\ket{\Psi}$ and $-\ket{\Psi}$ are physically indistinguishable because they differ only by a global phase.  Controlled access to a fixed implementation of $U_S$ supplies the phase reference that makes \eqref{eq:H44} operational.

A stronger optional access model can separate the two signs globally.  If a phase-aligned magnitude-state circuit is also available,
\begin{equation}
U_{|u|}\ket{0^n}
=
\ket{|\Psi|}
:=
\sum_i|u_i|\ket{i},
\label{eq:H48}
\end{equation}
then controlled preparation and one Hadamard gate give
\begin{equation}
\frac12\ket{0}(\ket{|\Psi|}+\ket{\Psi})
+
\frac12\ket{1}(\ket{|\Psi|}-\ket{\Psi}).
\label{eq:H49}
\end{equation}
The $a=0$ branch has amplitudes $u_i$ on positive entries and zero on negative entries, while the $a=1$ branch has amplitudes $|u_i|$ on negative entries and zero on positive entries.  Their success probabilities are exactly
\begin{equation}
P_+
=
\sum_{u_i>0}u_i^2,
\qquad
P_-
=
\sum_{u_i<0}u_i^2.
\label{eq:H50}
\end{equation}
This provides constant overhead when the selected sign mass is $\Omega(1)$, or inverse-polylogarithmic overhead when that mass is inverse polylogarithmic.  Because the principal access model does not assume $U_{|u|}$, the candidate-first interference test in ~\eqref{eq:H40}--\eqref{eq:H47} remains the default certification procedure.

\subsection{Four extrema objectives}
\label{app:four_extrema}

The main theorem is stated for the largest-magnitude extremum. The same construction also gives the largest positive value and the most negative value by selecting the corresponding tail of the amplitude distribution. The smallest-magnitude problem requires one additional preprocessing step because a small amplitude is not itself an extreme tail event.
\ssb{In every case, the relevant candidate tail is governed by the
same structural parameters as Theorem~1: its candidate count,
retained probability mass, transition width, and final
certification precision. Polylogarithmic complexity follows when
these parameters have the corresponding polylogarithmic scaling
with $N$.}

\begin{corollary}[Largest absolute value]
\label{cor:max_abs}
Let
\begin{equation}
i_{\rm abs}\in\arg\max_{i\in[N]}|u_i|.
\label{eq:ext_abs}
\end{equation}
Under the hypotheses of Theorem 1 of \textit{Main Text}, $i_{\rm abs}$ is obtained with the complexity stated there.
\end{corollary}

\begin{proof}
This is the canonical problem of Theorem 1 of \textit{Main Text}. The bounded two-sided filter constructed in Appendix~\ref{sec:polynomial_filter} suppresses the central background while retaining both sufficiently positive and sufficiently negative amplitudes. Equivalently, the two tails may be processed in separate runs. In either case their union contains every safely retained large-magnitude candidate, after which the value-certification procedure selects the largest $|u_i|$. Processing two tails separately changes only a constant prefactor.
\end{proof}

\begin{corollary}[Largest positive value]
\label{cor:max_positive}
Assume that the positive tail contains at least one safely retained amplitude and define
\begin{equation}
i_{+}\in\arg\max_{i:\,u_i>0}u_i.
\label{eq:ext_positive}
\end{equation}
If the positive-tail candidate count and retained probability mass satisfy the corresponding conditions of Theorem 1, then $i_{+}$ is found with the same polylogarithmic dependence on $N$.
\end{corollary}

\begin{proof}
The one-sided polynomial constructed in Appendix~\ref{sec:polynomial_filter} can be placed on the positive tail. Let $p_{+,d}$ satisfy $|p_{+,d}(x)|\leq\eta_{\rm filt}$ below the positive candidate threshold and $|1-p_{+,d}(x)|\leq\eta_{\rm filt}$ on the safely retained positive region. The QG1 transformation then gives
\begin{equation}
p_{+,d}(D_\Psi)|\Psi\rangle=\sum_i p_{+,d}(u_i)u_i|i\rangle .
\label{eq:positive_filter}
\end{equation}
Negative amplitudes and the central background are suppressed, while the safely retained positive amplitudes preserve their original magnitudes up to the uniform filter error. The filter is not required to determine which surviving positive amplitude is largest. It only has to retain $i_{+}$ in the candidate set. The original candidate amplitudes are then estimated and ordered during the certification step.

The same degree bound $d=O[\Gamma_N^{-1}\log(1/\eta_{\rm filt})]$ applies. 
\ssb{Hence the positive-tail problem obeys the same general complexity
bound as Theorem~1, with $K$ and $P_R$ defined for the positive
candidate tail. In particular, the scaling is polylogarithmic
whenever these structural parameters are polylogarithmically
controlled.}
\end{proof}

\begin{corollary}[Most negative value]
\label{cor:min_signed}
Define the signed minimum
\begin{equation}
i_{-}\in\arg\min_{i\in[N]}u_i=\arg\max_{i:\,u_i<0}(-u_i).
\label{eq:ext_negative}
\end{equation}
If the negative tail satisfies the corresponding structural conditions of Theorem 1 of the \textit{Main Text}, then $i_{-}$ is obtained with the same asymptotic complexity as the largest positive value.
\end{corollary}

\begin{proof}
The signed minimum is converted into a positive-tail maximum by reflecting the amplitude observable,
\begin{equation}
D_\Psi\longmapsto-D_\Psi.
\end{equation}
If $p_{+,d}(x)=\sum_{\ell=0}^{d}c_\ell T_\ell(x)$ is the positive-tail filter, the corresponding negative-tail filter is
\begin{equation}
p_{-,d}(x):=p_{+,d}(-x)=\sum_{\ell=0}^{d}(-1)^\ell c_\ell T_\ell(x),
\label{eq:negative_filter}
\end{equation}
because $T_\ell(-x)=(-1)^\ell T_\ell(x)$. Thus no new Chebyshev-walk construction is required. Reflection only changes the signs of the odd Chebyshev coefficients, leaving both the degree $d$ and coefficient normalization $C_d=\sum_\ell|c_\ell|$ unchanged.

The safely retained negative amplitudes are therefore converted into the same candidate problem as a positive maximum of $-u_i$. Candidate recovery and final certification then return the most negative original amplitude. If a sign check is required, it is performed using the candidate-localized interference procedure in Sec.~\ref{subsec:sign estimation}.
\end{proof}

The first three objectives are therefore different views of the same amplitude-tail problem. The largest magnitude uses both tails, the positive maximum uses the upper tail, and the signed minimum uses the reflected lower tail. No additional state representation is required.

The smallest-magnitude problem is different. A value close to zero belongs to the center of the amplitude distribution rather than to an extreme tail, so QG1 cannot expose it simply by reversing the sign of the filter. A common background shift converts this central minimum into an ordinary tail extremum.

\begin{corollary}[Smallest absolute value by background subtraction]
\label{cor:min_abs}
Let
\begin{equation}
i_{\min |\,\cdot\,|}\in\arg\min_{i\in[N]}|u_i|.
\label{eq:ext_min_abs}
\end{equation}
Define $a_i:=|u_i|$ and let $b$ be an index-independent reference level. Suppose that the normalized background-subtracted magnitude state
\begin{equation}
|\Psi_{\rm bg}\rangle=\sum_{i=0}^{N-1}\widetilde r_i|i\rangle,
\qquad
\widetilde r_i:=\frac{r_i}{\|r\|_2},
\qquad
r_i:=a_i-b,
\label{eq:bg_state}
\end{equation}
can be prepared with $\polylog{N}$ overhead. If the relevant negative tail of $|\Psi_{\rm bg}\rangle$ satisfies the same candidate-count, retained-mass, and finite-transition-width conditions as Theorem 1 of \textit{Main Text}, then $i_{\min |\,\cdot\,|}$ is recovered with polylogarithmic query complexity.
\end{corollary}

\begin{proof}
Subtracting the same scalar $b$ from every magnitude preserves their ordering,
\begin{equation}
|u_i|<|u_j|
\quad\Longleftrightarrow\quad
r_i<r_j.
\label{eq:bg_order}
\end{equation}
Consequently,
\begin{equation}
\arg\min_i|u_i|
=
\arg\min_i r_i
=
\arg\max_i(-r_i).
\label{eq:bg_reduction}
\end{equation}
Thus the smallest magnitude becomes the most negative residual, or equivalently the largest positive amplitude after reflection.

Writing
\begin{equation}
D_{\rm bg}:=\sum_i\widetilde r_i|i\rangle\langle i|,
\end{equation}
the corresponding negative-tail QG1 branch is
\begin{equation}
p_{-,d}(D_{\rm bg})|\Psi_{\rm bg}\rangle
=
\sum_i p_{-,d}(\widetilde r_i)\widetilde r_i|i\rangle .
\label{eq:bg_filter}
\end{equation}
The polynomial suppresses the central and opposite-sign residuals while preserving the safely separated negative tail. Once the minimum-magnitude index is contained in the recovered candidate set, the original amplitudes $|u_i|$ are used for final certification, so the background subtraction does not alter the reported physical objective.

The reflection $p_{-,d}(x)=p_{+,d}(-x)$ changes neither the degree nor the coefficient normalization. The only additional cost relative to the canonical extrema problem is the preparation of $|\Psi_{\rm bg}\rangle$. Hence the overall scaling remains polylogarithmic when this preparation requires only $\polylog{N}$ overhead and the background-subtracted tail satisfies the same structural conditions as Theorem 1 of \textit{Main Text}.
\end{proof}

\begin{remark}[Background choice and access cost]
\label{rem:bg_access}
The reference level $b$ in \eqref{eq:bg_state} must be common to all indices if the objective is the exact global minimum in \eqref{eq:ext_min_abs}. A position-dependent local background $b_i$ would instead produce the objective $\arg\min_i(|u_i|-b_i)$, which identifies the deepest local deficit and need not coincide with $\arg\min_i|u_i|$. Moreover, efficient preparation of the state in \eqref{eq:bg_state} is an additional access requirement. It is not assumed to follow automatically from arbitrary signed amplitude access to $|\Psi\rangle$.
\end{remark}

\begin{corollary}[Complexity of the four structured extrema objectives]
\label{cor:four_extrema}
For the largest magnitude, largest positive value, and most negative value, suppose the relevant candidate tail satisfies the candidate-count and retained-mass conditions of Theorem 1 of \textit{Main Text}, the transition width is at least inverse polylogarithmic, candidate recovery uses $\polylog{N}$ successful filtered states, and the final value or sign precision is inverse polylogarithmic. Each objective is then solved with overall $\polylog{N}$ query complexity.

The same conclusion holds for the smallest-magnitude objective provided, in addition, that the background-subtracted magnitude state in ~\eqref{eq:bg_state} can be prepared with $\polylog{N}$ overhead and its relevant tail satisfies the same structural conditions.
\end{corollary}

\begin{proof}
Corollary~\ref{cor:max_abs} is the canonical theorem. Corollary~\ref{cor:max_positive} changes only the selected tail. Corollary~\ref{cor:min_signed} reflects the amplitude observable, which changes the signs of the odd Chebyshev coefficients but leaves $d$ and $C_d$ unchanged. Therefore the candidate-state preparation and certification costs have the same asymptotic scaling as Theorem~ 1 of \textit{Main Text}.
\end{proof}

For Corollary~\ref{cor:min_abs}, the background subtraction in \eqref{eq:bg_reduction} converts the minimum into the same opposite-tail maximum problem. Once $|\Psi_{\rm bg}\rangle$ is efficiently available, the identical Chebyshev-walk filtering, candidate recovery, and certification arguments apply. The additional state-preparation overhead is polylogarithmic by assumption and therefore does not alter the overall polylogarithmic dependence on $N$.

\section{Proof of global completeness}
\label{app:global_completeness}

\ssb{We prove Lemma~2 of the Main Text for general candidate count
and retained mass. Let $K:=|C_\Gamma|$,
$P_R:=\sum_{i\in R_\Gamma}|u_i|^2>0$, and choose
$i^\star\in\arg\max_{i\in[N]}|u_i|$ with
$p_\star:=|u_{i^\star}|^2$.
Since $P_R>0$, the safely retained set is nonempty. The global
maximum has magnitude at least as large as every member of this
set, so $i^\star\in R_\Gamma\subseteq C_\Gamma$.}

\ssb{Conditioned on successful QG1 preparation, the normalized
filtered state is
\begin{equation}
|\Psi_F\rangle
=
\frac{1}{\sqrt{P_e}}
\sum_i p_d(u_i)u_i|i\rangle,
\qquad
P_e
=
\sum_i|p_d(u_i)|^2|u_i|^2 .
\end{equation}}

\ssb{Because $|p_d|\leq1$ and
$|p_d(u_i)|\leq\eta_{\rm filt}$ outside $C_\Gamma$,
$P_e\leq Kp_\star+\eta_{\rm filt}^2$.
Also $P_R\leq Kp_\star$, and hence
$\eta_{\rm filt}^2
\leq(\eta_{\rm filt}^2/P_R)Kp_\star$.
Therefore
\begin{equation}
P_e
\leq
Kp_\star
\left(
1+\frac{\eta_{\rm filt}^2}{P_R}
\right).
\end{equation}}

\ssb{Since $i^\star\in R_\Gamma$,
$|p_d(u_{i^\star})|\geq1-\eta_{\rm filt}$.
The probability of observing a global-extremum address in one
successful filtered-state measurement therefore satisfies
\begin{equation}
q_\star
\geq
\frac{(1-\eta_{\rm filt})^2}
{K\left(1+\eta_{\rm filt}^2/P_R\right)}.
\end{equation}}

\ssb{After $M$ independent successful measurements,
$(1-q_\star)^M\leq e^{-Mq_\star}$. Hence
\begin{equation}
M_{\rm glob}
\geq
\frac{K}{(1-\eta_{\rm filt})^2}
\left(
1+\frac{\eta_{\rm filt}^2}{P_R}
\right)
\log\frac{1}{\delta_{\rm glob}}
\end{equation}
is sufficient to recover a global-extremum address with
probability at least $1-\delta_{\rm glob}$.}

\ssb{Conditional on this event, the final value-certification step
fails with probability at most $\delta_{\rm cert}$, so
\begin{equation}
\Pr\!\left[
\exists j\in[N]\;:\;
|u_j|>|u_{\hat{i}}|
\right]
\leq
\delta_{\rm glob}+\delta_{\rm cert}.
\end{equation}
This proves Lemma~2.}

\ssb{For the inverse-polylogarithmic specialization
$K\leq\kappa_{\rm cand}L$,
$P_R\geq c_PL^{-a}$, and
$\eta_{\rm filt}^2\leq\xi P_R$, with
$L=\log_2N$ and $0<\xi<1$, one obtains
$q_\star=\Omega(L^{-1})$ and therefore
$M_{\rm glob}=O[L\log(1/\delta_{\rm glob})]$.
The retained-mass exponent $a$ controls the cost of preparing
each filtered state, whereas the candidate count controls the
number of successful filtered states required for global recovery.}

\section{Proof of the two levels of classical verification}
\label{app:verification}

Proposition 1 distinguishes two questions. Level I verifies the reported extremum assuming faithful execution of the prescribed quantum circuit. Level II removes this assumption and verifies the quantum computation itself.

\subsection{Level I: verification of the reported extremum}
\label{app:verification_level1}

The successful filtered state is
\begin{equation}
|\Psi_F\rangle=\frac{1}{\sqrt{P_e}}\sum_i p_d(u_i)u_i|i\rangle,
\qquad
q_i=\frac{|p_d(u_i)|^2|u_i|^2}{P_e},
\label{eq:verification_filtered_state}
\end{equation}
where $q_i$ is the probability of observing address $i$. Let $\widehat C$ be the recovered candidate set and let $\hat i$ denote the candidate with the largest certified amplitude. Define the total filtered probability outside the recovered set as
\begin{equation}
R_{\rm miss}:=1-\sum_{i\in\widehat C}q_i=\sum_{j\notin\widehat C}q_j.
\label{eq:verification_Rmiss}
\end{equation}

Suppose that $\hat i$ lies in the safely retained region, so $|u_{\hat i}|\geq\lambda+\Gamma_N$. Assume, for contradiction, that an unobserved address $j\notin\widehat C$ satisfies $|u_j|>|u_{\hat i}|$. Then $j$ also lies in the safely retained region and therefore $|p_d(u_j)|\geq1-\eta_{\rm filt}$. Hence
\begin{equation}
q_j=\frac{|p_d(u_j)|^2|u_j|^2}{P_e}>\frac{(1-\eta_{\rm filt})^2|u_{\hat i}|^2}{P_e}.
\label{eq:verification_qj}
\end{equation}
Since the filter is bounded by one,
\begin{equation}
q_{\hat i}=\frac{|p_d(u_{\hat i})|^2|u_{\hat i}|^2}{P_e}\leq\frac{|u_{\hat i}|^2}{P_e}.
\end{equation}
Therefore any unobserved address capable of exceeding the reported extremum must satisfy
\begin{equation}
q_j>(1-\eta_{\rm filt})^2q_{\hat i}.
\label{eq:verification_unseen}
\end{equation}
But every unobserved probability satisfies $q_j\leq R_{\rm miss}$. Consequently,
\begin{equation}
R_{\rm miss}<(1-\eta_{\rm filt})^2q_{\hat i}
\label{eq:verification_level1}
\end{equation}
rules out every unobserved amplitude larger than $|u_{\hat i}|$. Since $\hat i$ has already been certified as the largest member of $\widehat C$, it is then the global extremum.

The quantum sampling and probability estimation needed to obtain $q_i$ and $R_{\rm miss}$ are part of the preceding quantum computation. Level I concerns the classical work required once these estimates, together with their confidence intervals, have been supplied under faithful circuit execution. The verifier first confirms from the candidate-value estimate that $\hat i$ lies in the safely retained region and then evaluates the sum in \eqref{eq:verification_Rmiss}. Since only $O(\log N)$ candidate addresses remain, this final classical certificate check requires $O(\log N)$ arithmetic operations.
With finite estimation error, let $R_{\rm miss}^{U}$ denote an upper confidence bound for the missing probability and let $q_{\hat i}^{L}$ denote a lower confidence bound for the probability of the returned extremum. The robust certificate is
\begin{equation}
R_{\rm miss}^{U}<(1-\eta_{\rm filt})^2q_{\hat i}^{L}.
\label{eq:verification_level1_robust}
\end{equation}
This proves Level I of Proposition 1. The faithful-execution assumption is essential. A device that deliberately fabricates its measurement record is not covered by this level.

\subsection{Level II: verification of an untrusted quantum computer output}
\label{app:verification_level2}

We now remove the assumption that the quantum device executes the prescribed circuit faithfully. For a proposed extremal address $\hat i$, define a binary quantum computation $\mathcal A_{\hat i}$ that executes the complete quantum extrema-search procedure and outputs $1$ if and only if its certified result is $\hat i$. For a unique extremum, or after fixing a deterministic tie-breaking rule, the correctness and global-completeness guarantees give
\begin{equation}
\Pr[\mathcal A_{i^\star}=1]\geq1-\delta,
\qquad
\Pr[\mathcal A_j=1]\leq\delta
\quad (j\neq i^\star).
\label{eq:verification_decision}
\end{equation}
Reducing $\delta$ below $1/3$ gives the standard bounded-error separation. A correct proposed address is accepted with probability greater than $2/3$, whereas an incorrect address is accepted with probability less than $1/3$. The extrema-verification task therefore has the binary decision form required by the classical verification protocol of ref.~\cite{mahadev2018classical}.

For intuition, write the decision circuit as
\begin{equation}
\mathcal A_{\hat i}=U_TU_{T-1}\cdots U_1,
\qquad
|\psi_t\rangle=U_tU_{t-1}\cdots U_1|\psi_0\rangle.
\label{eq:verification_circuit}
\end{equation}
The complete computation can be represented by the history state
\begin{equation}
|\Phi_{\hat i}\rangle=\frac{1}{\sqrt{T+1}}\sum_{t=0}^{T}|t\rangle_{\rm clock}|\psi_t\rangle,
\label{eq:verification_history}
\end{equation}
where $|t\rangle_{\rm clock}$ labels the stage of the gate sequence. The standard circuit-to-Hamiltonian construction associates this history with a local Hamiltonian whose terms check the input state, the successive relations $|\psi_{t+1}\rangle=U_{t+1}|\psi_t\rangle$, the clock register, and the final accepting output. A valid accepting computation corresponds to the low-energy case, so the verification can be viewed as estimating local terms of this computation-history Hamiltonian rather than reproducing all intermediate quantum states classically.

An untrusted quantum device could otherwise fabricate the measurement outcomes needed for this energy test. The protocol of ref.~\cite{mahadev2018classical} prevents this by requiring the device to commit to information associated with its quantum measurements before learning how those commitments will be tested. The classical verifier retains secret information that allows the subsequent measurement responses to be checked. The detailed cryptographic functions, trapdoor construction, and soundness proof are established in ref.~\cite{mahadev2018classical} and are not reproduced here.

\eqref{eq:verification_decision} is the only additional reduction needed for the present problem. It establishes that the claimed extremum is a bounded-error quantum decision computation of the form verified in ref.~\cite{mahadev2018classical}. Let $T_{\rm circ}$ denote the number of gates in this complete decision circuit and let $\kappa_{\rm sec}$ denote the cryptographic security parameter. The corresponding classical-verifier complexity obeys
\begin{equation}
T_V=\operatorname{poly}(T_{\rm circ},\kappa_{\rm sec}).
\label{eq:verification_level2_cost}
\end{equation}
Consequently, if $T_{\rm circ}=\polylog{N}$ and $\kappa_{\rm sec}=\polylog{N}$, then $T_V=\polylog{N}$. This proves Level II of Proposition 1 of the \textit{Main Text}.

The two verification levels have distinct meanings. Level I checks the extrema result itself using the compressed filtered state and assumes faithful quantum execution. Level II verifies the computation that produced the result and does not require the quantum device itself to be trusted. The cryptographic layer is external to the quantum extrema-search algorithm and is not required for Theorem 1 or Lemma 2 of the \textit{Main Text}.

\section{Structure factor, Reynolds-number scaling, and crossover estimates}
\label{app:structure_crossover}

Let $g_i$ denote the physical signal and
\begin{equation}
u_i=\frac{g_i}{\sqrt{N}\,g_{\rm rms}},
\qquad
g_{\rm rms}^2=\frac{1}{N}\sum_{i=0}^{N-1}g_i^2.
\label{eq:I1}
\end{equation}
\ssb{Order the amplitudes by magnitude,
$|g_{(1)}|\geq|g_{(2)}|\geq\cdots$, and let
$K_R:=|R_\Gamma|$ denote the number of safely retained
amplitudes. Define
\begin{equation}
\begin{aligned}
S_{\rm in}&:=\frac{g_{\rm pk}}{g_{\rm rms}},
&\qquad g_{\rm pk}&:=|g_{(1)}|,\\
\rho_R^2&:=
\frac{1}{K_R}
\sum_{j=1}^{K_R}
\left(\frac{|g_{(j)}|}{g_{\rm pk}}\right)^2 .
\end{aligned}
\end{equation}
Because $R_\Gamma$ is defined by a magnitude threshold, it
contains the $K_R$ largest-magnitude amplitudes, up to ties.
Its normalized probability mass is therefore
\begin{equation}
P_R
=
\sum_{i\in R_\Gamma}|u_i|^2
=
\frac{K_R\rho_R^2S_{\rm in}^2}{N}.
\end{equation}}

\ssb{For the inverse-polylogarithmic family used to obtain an explicit
scaling in $N$, let $P_R\geq c_PL^{-a}$, where
$L:=\log_2N$ and $a>0$. It follows that
\begin{equation}
S_{\rm in}
\geq
S_\star
:=
\frac{
\sqrt{c_PN}
}{
\rho_R\sqrt{K_R}\,L^{a/2}
}.
\label{eq:I4}
\end{equation}
Thus the physical structure condition entering Theorem~1 is
equivalently a requirement that the safely retained events carry
sufficient intensity relative to the rms background. The numerical
study uses the representative case $a=1$ and $c_P=1$.}

For a family of physically resolved problems, write
\begin{equation}
N_{\rm db}({\rm Re})=C_N{\rm Re}^{\alpha},
\qquad
S_{\rm in}({\rm Re})=C_S{\rm Re}^{\beta}.
\label{eq:I5}
\end{equation}
Substitution into \eqref{eq:I4} gives

\begin{equation}
C_S\mathrm{Re}^{\beta}
\geq
\frac{
\sqrt{c_PC_N}\,
\mathrm{Re}^{\alpha/2}
}{
\rho_R\sqrt{K_R}\,
[\log_2(C_N\mathrm{Re}^{\alpha})]^{a/2}
}.
\label{eq:I6}
\end{equation}
The exponent $\alpha$ depends on the physical problem and on the dimensionality of the data being searched. For one-dimensional Burgers flow, resolving the viscous layer gives the representative scaling $N_{\rm db}=O({\rm Re})$~\cite{bharadwaj2025homotopy}. For three-dimensional turbulence at Kolmogorov resolution, the corresponding line, slice, and volume counts scale as $O({\rm Re}^{3/4})$, $O({\rm Re}^{3/2})$, and $O({\rm Re}^{9/4})$, respectively.

It is useful to separate these physical scalings from the preparation-dominated algorithmic crossover in $N$. The following comparison concerns one successful amplitude-amplified QG1 candidate-state preparation, whose leading scaling is $\widetilde O[C_d(\log_2N)^2]$. It does not include the repeated-state factor $M_{\rm glob}$ required for complete candidate recovery and should therefore not be interpreted as the crossover of the full certified-search complexity. Suppressing constants, $C_d$, calibration, confidence, and final-precision factors, comparison with a classical scan is determined by
\begin{equation}
N\simeq(\log_2N)^2.
\label{eq:I7}
\end{equation}
The nontrivial integer-scale solution is $N=16$. An intentionally optimistic comparison with ideal Grover scaling is determined by
\begin{equation}
\sqrt{N}\simeq(\log_2N)^2.
\label{eq:I8}
\end{equation}
Writing $N=2^x$ gives $2^{x/2}=x^2$, whose large solution is $x=16$, or $N=2^{16}$. These values are asymptotic guideposts rather than finite-$N$ performance predictions because the suppressed constants and precision costs can be important.

For $C_N=1$, the corresponding nominal Reynolds-number values follow from ${\rm Re}=N^{1/\alpha}$. They are summarized in Table~\ref{tab:complexity_main} of the \textit{Main Text}. The ``ideal Grover'' column deliberately assumes the stronger comparison-oracle access model and is therefore a conservative benchmark for the amplitude-only setting studied in this work.

\section{End-to-end nonlinear-PDE-to-extremum complexity}
\label{sec:end_to_end}

For the fluid-dynamics application, the amplitude-state preparation is part of the algorithm rather than a free data-loading oracle. The quantum homotopy solver of Ref.~\cite{bharadwaj2025homotopy} prepares the normalized order-$M$ approximation to the nonlinear PDE solution,
\begin{equation}
U_H|0\rangle
=
\frac{u^{(M)}(T)}
{\|u^{(M)}(T)\|_2}
=:|\Psi(T)\rangle ,
\label{eq:J1}
\end{equation}
with stated complexity
\begin{equation}
Q_H
=
\widetilde O\!\left[
s_H\epsilon_{\rm LCU}
\log(N_gM^2\tau)
\operatorname{polylog}
(\epsilon_Q^{-1},\kappa_H)
\right],
\label{eq:J2}
\end{equation}
under the convergence, sparsity, conditioning, efficient-input, postselection, and accuracy assumptions of that solver~\cite{bharadwaj2025homotopy}. Here $N_g$ is the spatial grid size, $\tau$ is the number of time steps, $M$ is the homotopy order, $s_H$ is the lifted sparsity, $\kappa_H$ is the corresponding conditioning parameter, and $\epsilon_Q$ is the quantum-solver error.

If the desired extrema belong to a derived field such as a velocity-gradient component, the corresponding sparse derivative operator can be applied coherently before extrema search using the compact LCU matrix-operation construction of Ref.~\cite{bharadwaj2024compact}. We denote this additional cost by $Q_\nabla$. The complete scientific pipeline is therefore
\begin{equation}
\begin{aligned}
\text{nonlinear PDE}
&\xrightarrow{\;U_H\;}
|\Psi(T)\rangle\\
&\xrightarrow{\;U_\nabla\;}
|\Psi_\nabla(T)\rangle\\
&\xrightarrow{\;\text{extrema search}\;}
\hat i .
\end{aligned}
\label{eq:J3}
\end{equation}
No $N_g$-component classical reconstruction is performed between these stages.

For the fluid-dynamics application, one coherent invocation of the state-preparation module consists of nonlinear-PDE state preparation followed, when required, by the derived-field operation. We therefore define
\begin{equation}
Q_S:=Q_H+Q_{\nabla},
\label{eq:J4}
\end{equation}
where the homotopy solver has
\begin{equation}
Q_H
=
\widetilde O\!\left[
s_H\epsilon_{\rm LCU}
\log(N_gM^2\tau)
\operatorname{polylog}
(\epsilon_Q^{-1},\kappa_H)
\right]
\label{eq:J5}
\end{equation}
under its stated convergence, sparsity, conditioning, postselection, and accuracy assumptions~\cite{bharadwaj2025homotopy}, while $Q_{\nabla}=\polylog{N_g}$ under the compact-LCU assumptions of Ref.~\cite{bharadwaj2024compact}. Hence $Q_S=\polylog{N}$ for $N_g=\Theta(N)$ under these assumptions.

The extrema-search complexity $Q_{\rm ext}^{\rm AA}$ counts coherent calls to this complete preparation module. Combining the upstream preparation with Theorem~1 therefore gives
\begin{equation}
Q_{\rm e2e}
=
Q_SQ_{\rm ext}^{\rm AA}
=
\polylog{N},
\label{eq:J6_e2e}
\end{equation}
provided the structural, filtering, precision, and failure parameters satisfy the polylogarithmic conditions stated in Theorem~1. Thus the end-to-end PDE-to-extremum computation retains polylogarithmic query complexity in the resolved problem size. The corresponding gate complexity of the underlying homotopy and compact-LCU operations is likewise polylogarithmic under their stated assumptions~\cite{bharadwaj2025homotopy,bharadwaj2024compact}.

\subsection{Compatibility of the nonlinear regimes}
\label{subsec:Re_overlap}
The nonlinear-PDE and extrema-search stages impose different physical conditions and their parameters should not be identified numerically. For the normalized one-dimensional Burgers analysis of Ref.~\cite{bharadwaj2025homotopy},
\begin{equation}
R_H=O({\rm Re}),
\qquad
t_{\rm NS}
=
\frac{L}{\bar U}
\frac{2}{|\hat h|R_H},
\label{eq:J7_timewindow}
\end{equation}
so larger physical nonlinearity is admissible over a correspondingly shorter certified time interval. The reported calculations reach physically resolved $R_H=O(10^2)$ while maintaining convergence within this finite-time regime~\cite{bharadwaj2025homotopy}.

The extrema-search crossover estimate is instead obtained from the database scaling $N_{\rm db}({\rm Re})=C_N{\rm Re}^{\alpha}$. For a three-dimensional Kolmogorov-resolved volume, $\alpha=9/4$, and the nominal preparation-dominated comparison with ideal Grover scaling in \eqref{eq:I8} gives ${\rm Re}\simeq1.38\times10^2$ for $C_N=1$. The two Reynolds quantities are not the same, but their common $O(10^2)$ scale establishes an order-of-magnitude regime in which the independently derived nonlinear-state-preparation and extrema-extraction conditions are mutually compatible.

The $Re_\lambda\simeq433$ JHTDB signals used in the numerical study should therefore be interpreted separately. They provide a high-Reynolds-number stress test of the extrema-search stage using fully developed turbulent data~\cite{li2008public,yeung2012dissipation}; they are not a claim that the present homotopy calculation already prepares the identical three-dimensional turbulent state.

\subsection{Verification of the composite computation}

When $U_H$, any derived-field operation $U_\nabla$, and the extrema-search circuit all have efficient classical descriptions, the complete circuit in \eqref{eq:J3} can itself be reduced to the bounded-error decision computation used in Proposition~1 of the Main Text. The classical verification protocol of Ref.~\cite{mahadev2018classical} may therefore be applied to the complete PDE-to-extremum computation, under the assumptions of that protocol. This supplies an external verification layer for the end-to-end quantum execution in addition to the global extrema certificate proved internally by the present algorithm.

\section{Numerical definitions and data-processing pipeline}
\label{app:numerics}

The numerical study in the Main Text uses 576 real velocity-gradient line signals extracted from the forced-isotropic Johns Hopkins Turbulence Database simulation~\cite{li2008public}. Each signal contains $N=1024$ samples and is normalized as
\begin{equation}
u_i=\frac{g_i}{\|g\|_2}.
\label{eq:K1}
\end{equation}
The extraction geometry and physical database parameters are given in Sec.~\ref{sec:numerics} of the Main Text and are not repeated here.

For each normalized signal, let $|u_{(1)}|\ge|u_{(2)}|\ge\cdots$ denote the amplitudes ordered by magnitude. The numerical target is
\begin{equation}
P_{\rm tar}
=
\frac{1}{\log_2N},
\qquad
\Gamma_N
=
(\log_2N)^{-3/2}.
\label{eq:K2}
\end{equation}
The threshold scan identifies a candidate tail satisfying the prescribed mass criterion, after which a Jackson-damped Chebyshev approximation is evaluated over the tested polynomial degrees. The accepted degree $d_{\min}$ is the smallest tested degree satisfying the adopted retained-mass and concentration conditions. The Jackson damping used in these numerical filters is the same coefficient construction discussed in Appendix~\ref{sec:polynomial_filter}.

The principal diagnostics reported in Figs.~4--5 and Table~\ref{tab:fixedN} of the Main Text are
\begin{align}
P_K
&=
\sum_{j=1}^{K}|u_{(j)}|^2,
\\
P_e
&=
\sum_i|p_d(u_i)|^2|u_i|^2,
\\
p_{\rm succ}
&=
\frac{P_e}{C_d^2},
\\
q_K
&=
\frac{\sum_{j=1}^{K}|p_d(u_{(j)})|^2|u_{(j)}|^2}{P_e},
\\
K_{\rm fuzzy}
&=
|C_\Gamma|.
\label{eq:K3}
\end{align}
Thus $P_K$ tests whether the original extreme-event tail carries the required input mass, $P_e$ measures how much total mass remains after applying the polynomial, $p_{\rm succ}$ is the raw probability of obtaining the successful coefficient-weighted QG1 branch, $q_K$ measures the concentration of the normalized filtered state on the diagnostic leading set, and $K_{\rm fuzzy}$ is the number of amplitudes lying above the calibrated lower candidate threshold. The corresponding leakage is $1-q_K$.

We also report
\begin{equation}
P_{\max}^{\rm in}
=
|u_{(1)}|^2,
\qquad
P_{\max}^{\rm out}
=
\frac{|p_d(u_{(1)})|^2|u_{(1)}|^2}{P_e},
\label{eq:K4}
\end{equation}
whose ratio measures the increase in the conditional probability of observing the true maximum after successful filtering. This ratio is a sampling-concentration diagnostic and should not be confused with the complete circuit success probability $p_{\rm succ}$.

The finite-$N$ resource comparison applies the same amplitude-estimation precision rule to the present recovered candidate list and to the value comparison required by a GDH implementation under amplitude-only access. It therefore measures the finite-size benefit of performing value estimation only after candidate compression. It is not a comparison with ideal GDH supplied with a cost-free exact value-comparison oracle.

\ssb{Finally, the numerical results test one concrete
inverse-polylogarithmic specialization of the general bounds in
Theorem~1 at the fixed size $N=1024$. We use
$P_{\rm tar}=(\log_2N)^{-1}$ and
$\Gamma_N=(\log_2N)^{-3/2}$, corresponding to the
representative retained-mass exponent $a=1$. The tests measure
retained mass, candidate count, concentration of the successful
filtered state, coefficient normalization, and the finite
polynomial degree required by QG1. The asymptotic dependence on
$N$ follows from the analytical bounds rather than from this
fixed-$N$ ensemble.}

\bibliography{QMIN_ARXIV_REFS}

\end{document}